\documentclass[a4paper,11pt]{article}

\usepackage{amsthm}
\usepackage{fullpage}%
\usepackage{graphicx,tikz,enumerate}
\usepackage[colorlinks=true,citecolor=blue,linkcolor=red]{hyperref}
\usepackage[title]{appendix}
\usepackage{amssymb, amsfonts,mathtools,stmaryrd}
\usepackage{xcolor}
\usepackage{enumitem}
\usepackage{etoolbox}

\usepackage{tikz}
\usetikzlibrary{automata,arrows,shapes,snakes,automata,backgrounds,positioning,calc}
\usetikzlibrary{graphs}

\usepackage{algorithm}
\usepackage[noend]{algpseudocode}

\usepackage{dsfont}
\usepackage{bbm}
\newtheorem{definition}{\textbf{Definition}}

\newtheorem{theorem}[definition]{\textbf{Theorem}}

\newtheorem{lemma}[definition]{\textbf{Lemma}}

\newtheorem{proposition}[definition]{\textbf{Proposition}}

\newtheorem{example}[definition]{\textbf{Example}}

\newtheorem{dpb}[definition]{\textbf{Decision Problem}}
\newtheorem{mpb}{\textbf{Minimization problem}}
\newtheorem{task}[definition]{\textbf{Task}}

\newtheorem*{proof*}{\textbf{Proof}}

\newcommand{\N}{\ensuremath{\mathbb{N}}}
\newcommand{\Z}{\ensuremath{\mathbb{Z}}}

\newcommand{\A}{\ensuremath{\mathcal{A}}}
\newcommand{\Lan}{\ensuremath{\mathcal{L}}}
\newcommand{\msf}[1]{\ensuremath{\mathsf{#1}}}
\newcommand{\Pos}{\ensuremath{\mathcal{P}}}

\newcommand{\Inter}[1]{\ensuremath{\llbracket 1,#1 \rrbracket}}
\newcommand{\TInter}[1]{\ensuremath{\llbracket #1 \rrbracket}}
\newcommand{\Intr}{\ensuremath{\llbracket 1,r \rrbracket}}
\newcommand{\Ints}{\ensuremath{\llbracket 1,s \rrbracket}}

\newcommand{\qinit}{\ensuremath{q_{\mathsf{init}}}}

\usepackage{bookmark}

\usepackage{pifont}
\usepackage{authblk}
\usepackage{graphicx}
\usepackage{multirow}
\author[1,2,3]{Benjamin Bordais}
\author[1,2]{Daniel Neider}
\affil[1]{TU Dortmund University, Dortmund, Germany}
\affil[2]{Center for Trustworthy Data Science and Security, University Alliance Ruhr, Dortmund, Germany}
\affil[3]{Université Libre de Bruxelles, Belgium}

\date{}
\usepackage{array,booktabs}
\newcolumntype{M}[1]{>{\centering\arraybackslash}m{#1}}

\begin{document}
	\title{Learning DFAs from Positive Examples Only via Word Counting}
	\maketitle	
	
	\begin{abstract}
		Learning finite automata from positive examples has recently gained attention as a powerful approach for understanding, explaining, analyzing, and verifying black-box systems. The motivation for focusing solely on positive examples arises from the practical limitation that we can only observe what a system is capable of (positive examples) but not what it cannot do (negative examples). Unlike the classical problem of passive DFA learning with both positive and negative examples, which has been known to be NP-complete since the 1970s, the topic of learning DFAs exclusively from positive examples remains poorly understood. This paper introduces a novel perspective on this problem by leveraging the concept of counting the number of accepted words up to a carefully determined length. Our contributions are twofold. First, we prove that computing the minimal number of words up to this length accepted by DFAs of a given size that accept all positive examples is NP-complete, establishing that learning from positive examples alone is computationally demanding. Second, we propose a new learning algorithm with a better asymptotic runtime that the best-known bound for existing algorithms. While our experimental evaluation reveals that this algorithm under-performs state-of-the-art methods, it demonstrates significant potential as a preprocessing step to enhance existing approaches.
	\end{abstract}

\section{Introduction}
In recent years, the rapid advancement and proliferation of Artificial Intelligence systems have transformed numerous fields, from healthcare \cite{zhang2022applications} to transportation \cite{bharadiya2023artificial} or finance \cite{chen2023explainable}. This surge in AI capabilities has greatly increased the need of interpretable models able to provide insights into the decision-making processes of complex black-box systems. 

Interpretable models are synthesized from temporal 
data. The goal is to obtain a model that is coherent with the observations gathered, which in turn can be used to decipher the system's past behavior
. Different kinds of (interpretable) models are used in the literature, often of one of two types: finite-state machines \cite{weiss2024extracting} and (temporal) logic formulas \cite{DBLP:conf/aips/CamachoM19}. In this paper, we focus on Deterministic Finite Automata (DFAs) \cite{DBLP:journals/ibmrd/RabinS59}, that enjoy both an accessible formalism, making them suitable e.g. classifier of sequential data \cite{DBLP:conf/aaai/ShvoLIM21}, and many desirable properties, such as tractable model checking. 

In the literature, synthesizing a DFA from observations is usually formalized as a passive learning task where we are given a positive set of words that a prospective DFA should accept, a negative set of words that it should reject, and an integer bound on its number of states \cite{DBLP:journals/iandc/Gold78,DBLP:conf/cade/GrinchteinLP06,DBLP:conf/icgi/HeuleV10}. 
However, a significant obstacle on the usefulness of the passive learning task from positive and negative examples (PLPN) is the difficulty of coming up with negative examples: they correspond to observations of a system's undesirable behaviors, which can be intrinsically hard (e.g., with black-box systems) or unrealistic (e.g., with self-driving cars) to gather. 

To circumvent that issue, the focus has recently shifted 
to the passive learning task of synthesizing DFAs from positive examples only (PLP). Without any additional restriction, the DFA accepting exactly the (positive) words of the input set $\mathcal{P}$---on the one end of the spectrum---and the trivial accepting DFA (the one-state DFA accepting every word)---on the other end of the spectrum---both meet the requirement of this PLP task while providing no insight whatsoever w.r.t. the input set $\mathcal{P}$. To forbid these two extreme solutions, 
constraints on the prospective DFA are added. On the one hand, as in the classical PLPN setting, we consider a bound $n \in \N$ on the number of states of the prospective DFA. The benefit is twofold: first, it prevents synthesizing the DFA accepting exactly the words in $\mathcal{P}$ (assuming $n$ is small enough); second, for explainability purposes, the smaller the DFA, the better. On the other hand, we look for a DFA that is language-minimal, i.e., whose language is minimal (for inclusion) among those DFAs with at most $n$ states that accept all words in $\mathcal{P}$. As argued by 
Avellaneda and Petrenko (2018), this constraint follows naturally from a parsimony standpoint: we look for a simplest solution. With these restrictions, we obtain a PLP task $\mathcal{T}$.

A counterexample (CE) guided algorithm solving the task $\mathcal{T}$ was first developed by Avellaneda and Petrenko (2018), then a 
few years later it was improved into a symbolic algorithm by 
Roy et al. (2023), which now constitutes the state-of-the-art. The CE algorithm relies on iteratively synthesizing DFAs by solving instances of the PLPN task, where $\mathcal{P}$ always constitutes the positive set, while the negative set is initially empty and iteratively grown by adding 
words witnessing the non language-minimality of candidate DFAs (i.e., counterexamples). 
On the other hand, the symbolic algorithm starts from the trivial accepting DFA, and iteratively looks for DFAs accepting all words in $\mathcal{P}$ with a smaller language. Both the CE and symbolic algorithms iteratively call SAT solvers at each step to find new candidate DFAs. 

By construction, the symbolic algorithm 
avoids redundancies that could occur in the CE algorithm
. However, even with the symbolic algorithm, a significant drawback remains: the best bound on the number of steps of the algorithm, and thus on the number of calls to SAT solvers, is exponential in $n$. This is due to the fact that there are exponentially many words of length at most (polynomial in) $n$. Furthermore, while the complexity of the PLPN task (which can be stated as a decision problem) is well understood---it is known to be NP-complete since the 70s \cite{DBLP:journals/iandc/Gold78}---the PLP task is poorly understood complexity-wise. In particular,  \cite{DBLP:conf/icgi/AvellanedaP18,DBLP:conf/aaai/0002GBN0T23} do not provide a complexity lower bound on any closely-related decision problem. As such, it is not clear that solving the PLP task $\mathcal{T}$ even requires the use of SAT solvers.

\paragraph{Our contributions.} In this paper, we tackle the PLP task $\mathcal{T}$ from a new perspective which consists in assigning to each DFA a word-counting metric that greatly enhances our ability to compare two DFAs (as opposed to checking language inclusion). Specifically, this word-counting metric is the number of words accepted by each DFA up to length $2n-2$
. To demonstrate the usefulness of this new perspective, we show that a word-counting minimal DFA 
is necessarily language-minimal
. Thus, we focus on a new PLP task $\mathcal{T}'$---that requires a prospective DFA to be word-counting minimal instead of only language-minimal---and on the associated decision problem about the minimal word-counting metric of DFAs with at most $n$ states accepting all words in $\mathcal{P}$. Our main contribution is that this problem is $\msf{NP}$-complete, which constitutes the first complexity result directly related to a PLP task. 
This $\msf{NP}$-completeness result has two main consequences: first, because it is in $\msf{NP}$, we immediately obtain that the task $\mathcal{T}'$, and thus the task $\mathcal{T}$ as well, can be solved with polynomially many calls to a SAT solver, via a binary search; second, because the problem is $\msf{NP}$-hard, it is unlikely that we can do better than polynomially many calls to a SAT solver to solve task $\mathcal{T}'$ (though task $\mathcal{T}$ could be intrinsically easier to solve than task $\mathcal{T}'$). 
%

We have implemented an algorithm solving task $\mathcal{T}'$ (and thus also task $\mathcal{T}$) via an Integer Linear Programming (ILP) encoding. Our experiments show that this algorithm under-performs the symbolic state-of-the-art algorithm solving task $\mathcal{T}$
. However, building on the word-counting ideas presented above, we have developed a method that could be used to improve the performance of the symbolic algorithm for solving $\mathcal{T}$. Specifically, 
we have designed a heuristic algorithm that looks for a DFA with at most $n$ states, accepting all words in $\mathcal{P}$, with a word-counting metric as small as possible. This algorithm is intended to be used as a pre-processing step outputting a DFA that can then be used as a starting point by the symbolic algorithm
. Our experiments show promising results of this method
. 

Missing technical details can be found in the extended version \cite{bordais2025learningdfaspositiveexamples}.
%

\paragraph{Related work.} As already mentioned, there is a large body of work focusing on the classical DFA passive learning. Another widely studied DFA learning setting is active learning, where a learner successively queries a teacher. Angluin's seminal work \cite{DBLP:journals/iandc/Angluin87} essentially started the research on this topic, and is still very influential nowadays \cite{DBLP:conf/fm/ShahbazG09,DBLP:conf/ijcai/BolligHKL09}. 

DFA learning from positive examples only was introduced by Gold 
(1967), where it is shown that DFAs cannot be identified at the limit. This partly explains why there was no further study of this topic until recently \cite{DBLP:conf/icgi/AvellanedaP18,DBLP:conf/aaai/0002GBN0T23}. More popular subjects are the passive learning from positive examples only of better-behaved models, such as pattern languages \cite{DBLP:journals/jcss/Angluin80}, hidden Markov chains \cite{stolcke1992hidden}, or stochastic machines \cite{DBLP:journals/ita/CarrascoO99}.

Finally, an ILP encoding a learning task (of a logical (LTLf) formula) was first introduced in \cite{DBLP:conf/ilp/IeloLFRGR23}.


\section{Definitions and notations}
\label{sec:definition}
We let $\N$ denote the set of non-negative integers. For all $i \leq j \in \N$, we let $\TInter{i,j}$ denote the set $\{ i,i+1,\ldots,j \} \subseteq \N$. 

\paragraph{Alphabet and automata.}
An alphabet $\Sigma$ refers to a non-empty finite set of letters. We let $\Sigma^*$ (resp. $\Sigma^+$) denote the set of (resp. non-empty) finite words with letters in $\Sigma$. We let $\epsilon \in \Sigma^*$ denote the empty word. For all words $u \in \Sigma^*$, we let $|u| \in \N$ denote the length of the word $u$, i.e., its number of letters. 
For all $m \in \N$, we let $\Sigma^m$ (resp. $\Sigma^{\leq m}$) denote the set of words of length $m$ (resp. at most $m$). For all $L \subseteq \Sigma^*$, we let $\msf{Pref}(L)$ denote the set of (non-strict) prefixes of words in $L$.
%
%
%
%
Let us now recall the definition of DFAs.
\begin{definition}
	Consider an alphabet $\Sigma$. A \emph{deterministic finite automaton (DFA)} on $\Sigma$ is a tuple $\A = (Q,\Sigma,\qinit,\delta,F)$ where $Q$ is a finite non-empty set of states, $\qinit \in Q$ is the initial state, $\delta\colon Q \times \Sigma \to Q$ is the transition function, and $F \subseteq Q$ is the set of final states. The transition function $\delta$ is naturally extended to finite words into a function $\delta^*\colon Q \times \Sigma^* \to Q$ defined as follows: for all $q \in Q$, $\delta^*(q,\epsilon) := q$, and for all $u \in \Sigma^*$ and $\alpha \in \Sigma$: $\delta^*(q,u \cdot \alpha) := \delta(\delta^*(q,u),\alpha)$. 
	
	A word $u \in \Sigma^*$ is \emph{accepted} by the DFA $\A$ if $\delta^*(\qinit,u) \in F$. We let $\Lan(\A) \subseteq \Sigma^*$ denote the language of the DFA $\A$, i.e., the set of words that it accepts: $\Lan(\A) := \{ u \in \Sigma^* \mid \delta^*(\qinit,u) \in F \}$. We let $|\A|$ denote the size of the DFA $\A$, i.e., its number of states $|\A| := |Q|$. Furthermore, let $\msf{DFA}(\Sigma)$ denote the set of DFAs on the alphabet $\Sigma$.
\end{definition}

In the following, unless otherwise stated, a DFA $\A$ on an alphabet $\Sigma$ will refer to the tuple $\A = (Q,\Sigma,\qinit,\delta,F)$.

Given an integer $n \in \N$ and a set of words $\mathcal{P}$, we let $\msf{Rec}(\Pos,n) := \{ \A \in \msf{DFA}(\Sigma) \mid |\A| \leq n,\; \Pos \subseteq \Lan(\A) \}$ denote the set of DFAs with at most $n$ states that accept all words in $\mathcal{P}$.

\paragraph{Strict partial orders.}
A relation $\prec \subseteq \msf{DFA}(\Sigma) \times \msf{DFA}(\Sigma)$ is a strict partial order on $\msf{DFA}(\Sigma)$ if it is irreflexive, antisymmetric, and transitive. For $S \subseteq \msf{DFA}(\Sigma)$ and $\A \in S$, $\A$ is said to be $S$-minimal w.r.t. $\prec$ if, for all $\A' \in S$, we do \emph{not} have $\A' \prec \A$.

Given two DFAs $\A,\A' \in \msf{DFA}$, we write $\A \prec_\Lan \A'$ when $\Lan(\A) \subsetneq \Lan(\A')$. Clearly, $\prec_\Lan$ is a strict partial order on $\msf{DFA}(\Sigma)$. 

\paragraph{Learning from positive examples only.}
In this paper, we focus on the task of synthesizing a DFA with at most $n \geq 1$ states accepting all words in a given set $\mathcal{P}$, i.e., we look for DFAs in $\msf{Rec}(\mathcal{P},n)$. 
Furthermore, among those DFAs, we look for a language-minimal one (following the parsimony principle), i.e., one that is $\msf{Rec}(\mathcal{P},n)$-minimal w.r.t. the order $\prec_\Lan$. Formally, the task that we consider is defined below.
\begin{task}
	\label{task}
	Given as input an alphabet $\Sigma$, a set $\Pos \subseteq \Sigma^*$, and $n \geq 1$, find a DFA $\A$ such that: a) $\A \in \msf{Rec}(\Pos,n)$; and b) the DFA $\A$ is $\msf{Rec}(\Pos,n)$-minimal w.r.t. $\prec_\Lan$.
	
	The size of this task is equal to $n + |\msf{Pref}(\Pos)|$.
\end{task}

\section{Learning DFAs via word counting}
\label{sec:DFA_word_counting}
Before we detail our approach for solving Task~\ref{task}, let us describe the state-of-the-art (symbolic) algorithm.

\paragraph{The state-of-the-art algorithm.} It was developed by 
Roy et al. (2023) and proceeds as follows. Initially, one starts with the trivial accepting DFA $\A_0 \in \msf{Rec}(\Pos,n)$. Then, at step $i \geq 1$, one computes a DFA $\A_i \in \msf{Rec}(\Pos,n)$ satisfying $\A_i \prec_\Lan \A_{i-1}$ using a SAT solver. The process stops when no such DFA is found. Overall, this algorithm iteratively computes a decreasing chain of DFAs $\A_0 \succ_\Lan \ldots \succ_\Lan \A_k \succ_\Lan \ldots$ with $\A_i \in \msf{Rec}(\mathcal{P},n)$, for all $i \in \N$. Since the languages of DFAs may be infinite, it may not be clear that this 
chain is necessarily finite. However, it is indeed the case, as argued by 
Roy et al. (2023): two DFAs with at most $n$ states have the same language if they accept exactly the same words of length at most $n^2$ (a tighter bound exists, as discussed below). Thus, since there are finitely many words of length at most $n^2$, the above decreasing chain of DFAs is necessarily finite. However, the best bound on the number of iterations that we can extract from this argument is exponential in $n$ (since there are exponentially many words of length at most $n^2$), with each iteration involving a call to a SAT solver. 

\paragraph{From a language-inclusion to a numerical order.} To design a method that does significantly less calls to SAT solvers, we shift perspective from focusing on language inclusion, via the strict partial order $\prec_\Lan$, to focusing on a numerical order induced by values that we assign to DFAs. As we will see in the remainder of this paper, this has two main benefits: first, numerical orders are compatible with binary searches
; second, the values assigned to DFAs can be seen as a score that we aim at optimizing (in this case, minimize). However, it is crucial that this numerical order relates to the language-inclusion order $\prec_\Lan$ previously defined, as this is the one involved in Task~\ref{task}. In fact, taking as values assigned to DFAs the number of words, up to a certain length, that they accept does the job. This is formally defined below.
\begin{definition}
	Consider some $h \in \N$. For all DFAs $\A,\A'$, we write $\A \prec_h \A'$ when $|\Lan(\A) \cap \Sigma^{\leq h}| < |\Lan(\A') \cap \Sigma^{\leq h}|$.
\end{definition}
In other words, we have $\A \prec_h \A'$ when $\A$ accepts fewer words of length at most $h$ than $\A'$. Clearly, for all $h \in \N$, $\prec_h$ is a strict partial order on DFAs. Furthermore, this order satisfies the following property, thus showing that minimizing the number of accepted words also minimizes the language.
\begin{proposition}
	\label{prop:language-minimal}
	For an alphabet $\Sigma$, a set $\Pos \subseteq \Sigma^*$, $n \geq 1$, and $h := 2n-2$, a DFA $\A \in \msf{Rec}(\mathcal{P},n)$ that is $\msf{Rec}(\mathcal{P},n)$-minimal w.r.t. $\prec_h$ is $\msf{Rec}(\mathcal{P},n)$-minimal w.r.t. $\prec_\Lan$. 
\end{proposition}
This proposition relies on the folk theorem recalled below.
\begin{theorem}[Folk result]
	\label{thm:small_witness}
	Consider an alphabet $\Sigma$, and $\A,\A' \in \msf{DFA}(\Sigma)$. If $\Lan(\A) \neq \Lan(\A')$, then $\Sigma^{\leq |\A| + |\A'| - 2} \cap (\Lan(\A) \setminus \Lan(\A') \cup \Lan(\A') \setminus \Lan(\A)) \neq \emptyset$.
\end{theorem}
We provide a proof of this theorem in \cite{bordais2025learningdfaspositiveexamples}. It essentially relies on the iterative computation of the Myhill-Nerode equivalence classes of DFAs. We also show that the bound $h = |\A| + |\A'| - 2$ is tight. 

The proof of Proposition~\ref{prop:language-minimal} is now direct.
\begin{proof}
	Consider a DFA $\A \in \msf{Rec}(\Pos,n)$ that is $\msf{Rec}(\Pos,n)$-minimal w.r.t. $\prec_h$, for $h = 2n-2$. Consider any DFA $\A' \in \msf{Rec}(\Pos,n)$ and assume towards a contradiction that we have $\A' \prec_\Lan \A$, i.e., $\Lan(\A') \subsetneq \Lan(\A)$. By Theorem~\ref{thm:small_witness}, there is some $u \in \Sigma^{\leq h} \cap (\Lan(\A) \setminus \Lan(\A'))$. Since we do not have $\A' \prec_h \A$, it follows that $|\Lan(\A) \cap \Sigma^{\leq h}| \leq |\Lan(\A') \cap \Sigma^{\leq h}|$. Hence, there must be some $u' \in \Sigma^{\leq h} \cap (\Lan(\A') \setminus \Lan(\A))$. Thus, we do not have $\Lan(\A') \subseteq \Lan(\A)$. Hence, the contradiction. 
\end{proof}

Following Proposition~\ref{prop:language-minimal}, we consider a new task. 
\begin{task}
	\label{task:counting}
	Given as input an alphabet $\Sigma$, a set $\Pos \subseteq \Sigma^*$, and $n \geq 1$, find a DFA $\A$ such that: a) $\A \in \msf{Rec}(\Pos,n)$; and b) the DFA $\A$ is $\msf{Rec}(\Pos,n)$-minimal w.r.t. $\prec_{2n-2}$.
\end{task}
Proposition~\ref{prop:language-minimal} gives that a DFA satisfying the requirements of Task~\ref{task:counting} also satisfies the requirements of Task~\ref{task}. Therefore, let us focus on solving this Task~\ref{task:counting}.

\section{Computing the minimal number of accepted words of length at most $2n-2$.}
\label{sec:investigate_decision_problem}
To study how we can solve Task~\ref{task:counting}, we investigate the associated decision problem, formally defined below. 
\begin{dpb}
	\label{problem1}
	Given as input an alphabet $\Sigma$, a set $\Pos \subseteq \Sigma^*$, $n \geq 1$, and $k \in \N$, decide if there exists a DFA $\A \in \msf{Rec}(\Pos,n)$ such that: $|\Lan(\A) \cap \Sigma^{\leq 2n-2}| \leq k$.
	
	The size of the input of this task is $n + |\msf{Pref}(\Pos)| + \log k$
	. 
\end{dpb}
As a first step in the study of this problem, let us consider how we can compute the number of words, up to a certain length, accepted by a DFA. Actually, this can be done in a direct (iterative) manner, as described in Lemma~\ref{lem:iterative_computation_number_of_words} below.
\begin{lemma}
	\label{lem:iterative_computation_number_of_words}
	Consider an alphabet $\Sigma$ and a DFA $\A \in \msf{DFA}(\Sigma)$. For all $q \in Q$ and $m \in \N$, we let $\msf{N}(q,m) := \{ w \in \Sigma^m \mid \delta^*(q_0,w) = q \}$. Then, $\msf{N}(q,0) = 1$ if $q = q_0$, and $\msf{N}(q,0) = 0$ otherwise. Furthermore, for all $(m,q) \in \N \times Q$: $\msf{N}(q,m+1) = \sum_{q' \in Q} \msf{N}(q',m) \cdot |\{ \alpha \in \Sigma \mid \delta(q',\alpha) = q \}|$.
	
	In addition, $|\Lan(\A) \cap \Sigma^{\leq m}| = \sum_{i=0}^m \sum_{q \in F} \msf{N}(q,i)$.
\end{lemma}
Note that this lemma would not hold with automata that are not deterministic. 
%

Let us now state the main result of this paper. 
\begin{theorem}
	\label{thm:NP_complete}
	The decision problem~\ref{problem1} is $\msf{NP}$-complete.
\end{theorem}
\begin{proof}
	We argue the $\msf{NP}$-hardness in a dedicated section below. As for Problem~\ref{problem1} being in $\msf{NP}$, this problem can be solved as follows: we guess a DFA $\A \in \msf{Rec}(\mathcal{P},n)$---we can check that it is indeed the case in time polynomial in $|\msf{Pref}(\Pos)|$ by simulating the runs of $\A$ on words in $\mathcal{P}$---, we compute $|\Lan(\A) \cap \Sigma^{\leq 2n-2}| \in \N$ and compare it with $k$. Although $|\Lan(\A) \cap \Sigma^{\leq 2n-2}|$ may be exponential, we can compute it in polynomial time with the help of Lemma~\ref{lem:iterative_computation_number_of_words}. 
\end{proof}

Before we present the $\msf{NP}$-hardness proof, let us discuss the two major implications of Theorem~\ref{thm:NP_complete}. 
First, the fact that the decision problem~\ref{problem1} is in $\msf{NP}$ suggests an algorithm solving Task~\ref{task:counting}, and thus Task~\ref{task} as well, with seemingly better asymptotic guarantees than the state-of-the-art algorithm.
\begin{proposition}
	\label{thm:solveTask}
	Given an alphabet $\Sigma$, a set $\Pos \subseteq \Sigma^*$, and an integer $n \in \N$, we can solve Task~\ref{task:counting} with $\mathcal{O}(poly(|\Sigma|,n))$-calls to a SAT solver on inputs of size $\mathcal{O}(poly(|\Sigma|,|\mathsf{Pref}(\mathcal{P})|,n))$.
\end{proposition}
\begin{proof}	
	We have $|\Sigma^{\leq 2n-2}| = \frac{m^{2n-1} -1}{m-1}$, for $m := |\Sigma| \geq 2$. Hence, a binary search querying, at each step, a SAT solver on an encoding of Problem~\ref{problem1} with $1 \leq k \leq |\Sigma^{\leq 2n-2}|$ would take at most $\log(|\Sigma^{\leq 2n-2}|) \leq (2n-1) \cdot \log m$ steps to find a DFA $\A \in \msf{Rec}(\mathcal{P},n)$ satisfying $|\Lan(\A) \cap \Sigma^{\leq 2n-2}| = \min_{\A' \in \msf{Rec}(\mathcal{P},n)} |\Lan(\A') \cap \Sigma^{\leq 2n-2}|$. By definition, such a DFA $\A$ would be $\msf{Rec}(\mathcal{P},n)$-minimal w.r.t. $\prec_{2n-2}$.
\end{proof}

Second, unless $\msf{P} = \msf{NP}$, the $\msf{NP}$-hardness of Decision Problem~\ref{problem1} makes it unlikely that it is possible to solve Task~\ref{task:counting} with an asymptotic bound significantly better than that of Proposition~\ref{thm:solveTask}. This only applies to Task~\ref{task:counting}, and there might be an efficient way to solve Task~\ref{task} without solving Task~\ref{task:counting}. 

\paragraph{NP-hardness proof.} The NP-hardness proof of Theorem~\ref{thm:NP_complete} is quite technical, see \cite{bordais2025learningdfaspositiveexamples} for a detailed proof. 
Here, we only give a bird's eye view of the steps that we take to prove the result. 

The $\msf{NP}$-hardness proof relies on a reduction from the $\msf{All}$-$\msf{Pos}$-$\msf{Neg}$ $\msf{SAT}$ ($\msf{APN}$-$\msf{SAT}$) problem that was already used by 
Lingg, de Oliveira Oliveira, and Wolf (2024) to establish the $\msf{NP}$-hardness of the classical passive learning problem with positive and negative words (with a binary alphabet). The $\msf{APN}$-$\msf{SAT}$ problem is a variant of the classical $\msf{SAT}$ problem where we are given a set of clauses to satisfy, with each clause being constituted of only positive literals (i.e., without negations) or only negative literals (i.e., with negations). This problem is defined as follows.
\begin{dpb}
	\label{problem_sat_pos_neg}	
	Consider as input a set $X := \{ x_i \mid i \in \Inter{r} \}$ of variables, a set of clauses $\mathcal{C} = \{ C_j \mid j \in \Inter{s} \}$ that is described as the disjoint union of sets of positive and negative clauses: $\mathcal{C} = \mathcal{C}_+ \uplus \mathcal{C}_-$, where for all $j \in \Inter{s}$, we have $C_j \subseteq X$. We have to decide if there is a valuation $\nu\colon X \to \{\top,\bot\}$ of the variables satisfying $(X,\mathcal{C})$, i.e., such that, for all $j \in \Inter{s}$:
	\begin{align*}
		\exists x \in C_j:\; \nu(x) = \begin{cases}
			\top & \text{ if }C_j \in \mathcal{C}_+ \\
			\bot & \text{ if }C_j \in \mathcal{C}_- \\
		\end{cases}
	\end{align*}	
\end{dpb}

For the remainder of this section, we consider an instance $(X,\mathcal{C})$ of Problem~\ref{problem_sat_pos_neg} with $X = \{ x_i \mid i \in \Intr \}$ the set of variables, and $\mathcal{C} = \{ C_j \mid j \in \Ints \} = \mathcal{C}_+ \uplus \mathcal{C}_-$ the set of clauses. The alphabet that we consider is $\Sigma := \{a,b\}$. 
Our goal is to define a set of positive words $\mathcal{P}$ and two integers $n$ and $k$ such that $(X,\mathcal{C})$ is a positive instance of Problem~\ref{problem_sat_pos_neg} if and only if there is a $(\Pos,n,k)$-suitable DFA $\mathcal{A}$, i.e., such that 1) $\A \in \msf{Rec}(\mathcal{P},n)$; and 2) $|\Sigma^{\leq 2n-2} \cap (\Lan(\mathcal{A}) \setminus \Pos)| \leq k$\footnote{This condition constitutes a slight change compared to Problem~\ref{problem1}. This is however harmless, see the extended version.}. 

Let us first discuss what a DFA satisfying conditions 1) and 2) could look like on a simple input of Problem~\ref{problem_sat_pos_neg}. 
That way, we will be able to illustrate (some of) the kinds of words that we include in $\mathcal{P}$, which the DFA has to accept
.
\begin{example}
	\label{example:DFA}
	Consider the set of variables $X = \{ x_1,x_2,x_3 \}$, a single positive clause $C_1 = \{ x_1,x_3 \} \in \mathcal{C}_+$ and single a negative clause $C_2 = \{ x_2,x_3 \} \in \mathcal{C}_-$. In Figure~\ref{fig:smallDFA}, we scheme the shape of a DFA $\A_{\msf{ex}}$ that would be $(\mathcal{P},n,k)$-suitable, with $(\mathcal{P},n,k)$---which we partly describe here---being the result of the reduction from the instance $(X,\mathcal{C})$
	. The double-circled states are accepting. 
	
	In $\mathcal{P}$, we define a \textquotedblleft{}$\top$-word\textquotedblright{} $a \cdot a \cdot a \cdot a$ and a 
	\textquotedblleft{}$\bot$-word\textquotedblright{} $a \cdot b \cdot a \cdot a$. One can see that when the DFA $\A_{\msf{ex}}$ runs on those words, it visits the right-most (green and red) states. 
	
	We also define \textquotedblleft{}variables words\textquotedblright{} in $\mathcal{P}$ that the DFA $\A_{\msf{ex}}$ processes by visiting the central (blue) states. In particular, these words impose to $\A_{\msf{ex}}$ that we reach either $q_\top^1$ or $q_\bot^1$ upon reading the letter $a$ from the states $q_{x_i}'$, for $i \in \{1,2,3\}$. This corresponds to the bold transitions from the central states to the right-most states.
	These transitions uniquely define a valuation $\nu\colon X \to \{\top,\bot\}$, specifically: $\nu(x_1) = \top$, $\nu(x_2) = \bot$, $\nu(x_3) = \top$.
	
	In addition, we define \textquotedblleft{}clause words\textquotedblright{} in $\mathcal{P}$ that $\A_{\msf{ex}}$ processes by visiting the left-most (yellow) states. Some of these words ensure that upon reading the letter $b$ from $q_{C_j}$, for $j \in \{1,2\}$, we reach a state $q_{x_i}'$, for $i \in \{1,2,3\}$, such that $x_i \in C_j$. The other \textquotedblleft{}clause words\textquotedblright{} ensure that a positive (resp. negative) clause-state eventually leads to a $\top$-state (resp. $\bot$-state). Namely, the word $b \cdot b \cdot b \cdot a \in \mathcal{P}$ encodes that $C_1$ is a positive clause, and the word $b \cdot b \cdot a \cdot b \cdot a \cdot a \in \mathcal{P}$ encodes that $C_2$ is a negative clause.
	
	Overall, the valuation $\nu\colon X \to \{\top,\bot\}$ uniquely defined by $\A_{\msf{ex}}$ satisfies $(X,\mathcal{C})$, e.g., for the negative clause $C_2$, there is a bold edge from $q_{C_2}$ to $q'_{x_2}$, thus we must have $x_2 \in C_2$ and $\nu(x_2) = \bot$, which is indeed the case. 
\end{example}

\begin{figure}
	\centering
	\includegraphics[width=0.8\textwidth]{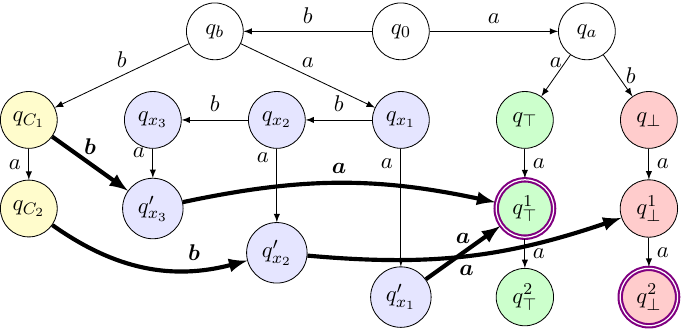}
	\caption{The shape of a DFA $\A_{\msf{ex}}$.
	}
	\label{fig:smallDFA}
\end{figure}

Let us now describe a little more formally how the words that we define in $\Pos$ ensure that the existence of a $(\Pos,n,k)$-suitable DFA $\mathcal{A}$ implies that $(X,\mathcal{C})$ is a positive instance of Problem~\ref{problem_sat_pos_neg}. Below, we consider a DFA $\mathcal{A} = (Q,\Sigma,\qinit,\delta,F)$ satisfying conditions 1) and 2). 
\begin{enumerate}
	\item We define two words $w_\top,w_\bot \in \mathcal{P}$ and we ensure that $\delta^*(q_0,w_\top) \neq \delta^*(q_0,w_\bot)$. This fact will allow us to distinguish when a variable is mapped to true ($\top$) and to false ($\bot$). In the DFA $\A_{\msf{ex}}$ of Example~\ref{example:DFA}, we have $\delta^*(q_0,w_\top) = q_\top^1 \neq q_\bot^1 = \delta^*(q_0,w_\bot)$.
	\item For each variable $x \in X$, we define a variable word $w_x \in \mathcal{P}$ and we ensure that, for all $x \neq x' \in X$, $\delta^*(q_0,w_x) \neq \delta^*(q_0,w_{x'})$. In the DFA $\A_{\msf{ex}}$, we have $\delta^*(q_0,w_{x_1}) = q_{x_1}'$, $\delta^*(q_0,w_{x_2}) = q_{x_2}'$, and $\delta^*(q_0,w_{x_3}) = q_{x_3}'$.
	\item We add words in $\mathcal{P}$ to ensure that there is some $w_{\msf{Var}} \in \Sigma^*$, such that, for all $x \in X$, we have $\delta^*(q_0,w_x \cdot w_{\msf{Var}}) \in \{ \delta^*(q_0,w_\top),\delta^*(q_0,w_\bot) \}$. In Example~\ref{example:DFA}, $w_{\msf{Var}} = a$. When this property is ensured, a valuation $\nu\colon X \to \{\top,\bot\}$ is uniquely defined by, for all $x \in X$: $\delta^*(q_0,w_x \cdot w_{\msf{Var}}) = \delta^*(q_0,w_{\nu(x)})$. (Note that the valuation $\nu\colon X \to \{\top,\bot\}$ is unique because $\delta^*(q_0,w_\top) \neq \delta^*(q_0,w_\bot)$.)
	\item For each clause $C \in \mathcal{C}$, we define a clause word $w_C \in \mathcal{P}$ and we ensure that, for all $C,C' \in \mathcal{C}$, $\delta^*(q_0,w_C) \neq \delta^*(q_0,w_{C'})$. In the DFA $\A_{\msf{ex}}$, we have $\delta^*(q_0,w_{C_1}) = q_{C_1}$, and $\delta^*(q_0,w_{C_2}) = q_{C_2}$.
	\item We add words in $\mathcal{P}$ to ensure that there is some $w_{\msf{Cl}} \in \Sigma^*$, such that, for all $C \in \mathcal{C}$, there is some variable $x_C \in C$ such that we have $\delta^*(q_0,w_C \cdot w_{\msf{Cl}}) = \delta^*(q_0,w_{x_C})$. In Example~\ref{example:DFA}, we have $w_{\msf{Cl}} = b$.
	\item Finally, we add words in $\mathcal{P}$ to ensure that, for all clauses $C \in \mathcal{C}$, we have: $\delta^*(q_0,w_C \cdot w_{\msf{Cl}} \cdot w_{\msf{Var}}) = \delta^*(q_0,w_\top)$ if $C \in \mathcal{C}_+$, and $\delta^*(q_0,w_C \cdot w_{\msf{Cl}} \cdot w_{\msf{Var}}) = \delta^*(q_0,w_\bot)$ if $C \in \mathcal{C}_-$. In Example~\ref{example:DFA}, this corresponds to the words $b \cdot b \cdot b \cdot a, b \cdot b \cdot a \cdot b \cdot a \cdot a \in \mathcal{P}$.
\end{enumerate}
With Items 5. and 6., we have that the valuation defined in Item 3. satisfies $(X,\mathcal{C})$. Indeed, consider e.g. a positive clause $C \in \mathcal{C}$. By Fact 5, there is variable $x_C \in C$ such that $\delta^*(q_0,w_C \cdot w_{\msf{Cl}}) = \delta^*(q_0,w_{x_C})$. By Fact 6, we have $\delta^*(q_0,w_{x_C} \cdot w_{\msf{Var}}) = \delta^*(q_0,w_C \cdot w_{\msf{Cl}} \cdot w_{\msf{Var}}) = \delta^*(q_0,w_\top)$. Thus, $\nu(x_C) = \top$ and the valuation $\nu$ satisfies the clause $C$. Note that, in the actual reduction described in the extended version \cite{bordais2025learningdfaspositiveexamples}, there are many words in $\mathcal{P}$ associated with $\top$, $\bot$, variables and clauses. 

Now that we have described the general idea of the construction, we would like to sketch how we prove that a DFA satisfying conditions 1) and 2) necessarily ensures the various properties laid out in Items 1., 2., 3., 4., 5., and 6. above. We use two different kinds of arguments that leverage the two opposing constraints on DFAs satisfying conditions 1) on 2): on the one hand, such DFAs must not make too many errors (i.e., accepting too many words not in $\mathcal{P}$), this is leveraged by TME-arguments (\textquotedblleft{}Too Many Errors\textquotedblright{}); on the other hand, such DFAs must not have too many states, this is leveraged by TMS-arguments (\textquotedblleft{}Too Many States\textquotedblright{}). 
\begin{itemize}
	\item TME-arguments are used to show that two words $w,w' \in \Sigma^*$ are mapped to two different states by the function $\delta^*(\qinit,\cdot)$. Typically, the structure of a TME-argument is as follows. Assume that there are $k+1$ words $(u_\sigma)_{\sigma \in \Inter{k+1}}$ such that we have $w \cdot u_\sigma \in \mathcal{P}$ and $w' \cdot u_\sigma \notin \mathcal{P}$ for all $\sigma \in \Inter{k+1}$. In that case, it cannot be that $\delta^*(\qinit,w) = \delta^*(\qinit,w')$, since otherwise we would have $w' \cdot u_\sigma \in \Lan(\mathcal{A}) \setminus \mathcal{P}$, for all $\sigma \in \Inter{k+1}$, which is a contradiction with condition 2) (assuming that $2n-2$ is large enough). This type of argument 
	allows to establish the properties described in Facts 1., 2. and 4.
	\item TMS-arguments are used to show that in two sets of words $W,W' \subseteq \Sigma^*$ there must be some $w \in W,w' \in W'$ such that $\delta^*(\qinit,w) = \delta^*(\qinit,w')$. To do so, it suffices to show that we have $|W| + |W'| > n$. 
	TMS-arguments 
	allow us to establish the properties described in Facts 3. and 5.
\end{itemize}
As for the properties of Fact 6, 
they are direct consequences of the properties of other facts (specifically, Fact 3 and Fact 5), and of the words that we define in $\mathcal{P}$.

\section{Experiments}
Our goal is now to use our new word-counting perspective to improve the state-of-the-art algorithm solving Task~\ref{task}. All of our code and data is available on GitHub (cf. the top page). Before we discuss our algorithms, let us 
describe the experimental setup that we have used.

\paragraph{Experimental setup.} We ran all of our experiments on Ubuntu 24.04.2 LTS, using 30 GiB of RAM, 12 CPU cores, and a clock speed of 5.2 GHz.

In order to compare our algorithms with the state-of-the-art symbolic algorithm developed by 
Roy et al. (2023), we have used the inputs used in that paper to compare the performance of the symbolic algorithm and the counterexample guided algorithm developed by Avellaneda and Petrenko (2018). These inputs are constituted of 28 samples of around 1000 words of lengths from 1 to 10. These words are obtained by sampling the languages of randomly generated small DFAs (of at most 10 states), with an alphabet of size 3. As in \cite{DBLP:conf/aaai/0002GBN0T23}, we consider a timeout (TO) of 1000s to solve instances of Task~\ref{task} and Task~\ref{task:counting}. All experiments are done with a number of states $n$ between 1 and 8. For all samples, we have run the algorithms 10 times, and we use the mean runtime (with the standard deviation).

\subsection{ILP algorithm solving Task~\ref{task:counting}}
To solve Task~\ref{task:counting}, we encode it as an Integer Linear Programming (ILP) minimization problem
. In that setting, the goal is to minimize a linear combination of integer variables, subject to a system of linear inequalities (with integer coefficients). When encoding Task~\ref{task:counting} in the ILP setting, the quantity to minimize corresponds to the number of words up to length $2n-2$ accepted by a prospective DFA. 

Let us describe on a high level how our ILP encoding works---the details can be found in the extended version \cite{bordais2025learningdfaspositiveexamples}. Our encoding relies on two types of variables: binary variables, encoding the DFA and checking that it is in $\msf{Rec}(\mathcal{P},n)$; and integer variables which are used to count the number of accepted words. The inequalities involving the binary variables are actually very close to the logical constraints used in the SAT-encoding developed by \cite{DBLP:conf/aaai/0002GBN0T23} (2023). As for the integer variables, our encoding relies on Lemma~\ref{lem:iterative_computation_number_of_words}: the inductive relation described in this lemma allows us to store in an integer variable the number of words of length at most $2n-2$ accepted by the prospective DFA, which we can set as a minimization goal, subject to the other inequalities. Note that, to avoid having variable multiplication, we had to use the classical Big-M method \cite{griva2008linear}.

We implemented our algorithm in Gurobi \cite{gurobi} which we have used with a Python3 API. In Figure~\ref{fig:compare_symb_ILP}, we have plotted in blue dots the computation time of our ILP algorithm (y-axis) and the computation time of the symbolic algorithm (x-axis) on each 28 samples with $n=4$. Clearly, the ILP algorithm largely under-performs the symbolic algorithm
. The situation is not significantly better for the case $n=6$ plotted as red stars\footnote{We only plotted the samples for which the ILP algorithm did not reach a full timeout (i.e., all 10 attempts timeout) before $n=6$.}. Overall, our ILP algorithm does not stand the comparison with the symbolic algorithm. We discuss below three possible reasons for this significant runtime difference.

\begin{figure}
	\centering
	\includegraphics[width=0.5\textwidth]{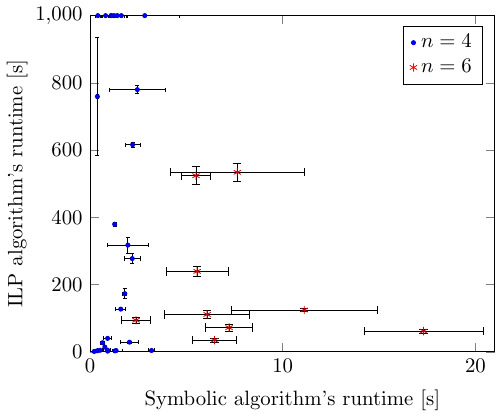}
	\caption{The comparison on the runtime of the symbolic and ILP algorithms.}
	\label{fig:compare_symb_ILP}
\end{figure}

First, from a theoretical standpoint, the symbolic algorithm could asymptotically take exponentially many steps to terminate. 
However, in practice, the symbolic algorithm does not take many steps to terminate (around 20, in the worst case), and thus greatly outperforms our ILP algorithm. There could be some kinds of samples for which the symbolic algorithm takes much more steps to conclude, in which case the runtime difference 
could be much smaller.

Furthermore, ILP solvers and SAT solvers internally use really different mechanisms: branch and bound for ILP solvers and DPLL for SAT solvers. It could be that in the case of these samples, DPLL works much better than branch and bound. This could be tested by implementing a binary search algorithm solving Task~\ref{task:counting} by querying SAT solvers. 

Finally, the ILP algorithm solves Task~\ref{task:counting}, while the symbolic algorithm solves Task~\ref{task}. We have shown in Proposition~\ref{prop:language-minimal} that Task~\ref{task:counting} is at least as hard to solve as Task~\ref{task}, it could very well be that it is simply much harder to solve.

\subsection{A pre-processing algorithm}
The symbolic algorithm starts from the trivial accepting DFA and then computes a $\prec_\Lan$-decreasing chain of DFAs until it finds a language-minimal one. To speed-up this process, we have designed a heuristic algorithm that computes a DFA that can then be used as a starting point for the symbolic algorithm. We first describe this heuristic algorithm, and then we discuss its experimental evaluation.

\paragraph{Description of the heuristic algorithm.} The goal of this heuristic algorithm is to output a DFA $\A \in \msf{Rec}(\mathcal{P},n)$ that accepts as few words of length at most $2n-2$ as possible. The general idea of the algorithm is to randomly search among DFAs, and pick one that minimizes the number of accepted words. However, it is not clear how to randomly search among DFAs in $\msf{Rec}(\mathcal{P},n)$. Therefore, our heuristic algorithm focuses on transition systems instead. A transition system can be seen as a DFA without starting or final states, i.e., it is described by a tuple $T = (Q,\Sigma,\delta)$ with $\delta\colon Q \times \Sigma \to Q$. Consider then a transition system $T = (Q,\Sigma,\delta)$ with $|Q| \leq n$ and a starting state $q \in Q$. It is clear that choosing $F_{T,q} := \{ \delta^*(q,w) \mid w \in \mathcal{P} \}$ minimizes the number of words accepted by the DFA $\A_{T,q} := (Q,\Sigma,\delta,q,F_{T,q})$ while ensuring that $\A_{T,q} \in \msf{Rec}(\mathcal{P},n)$. Following, given a transition system $T$, we define the \emph{score} of the transition system $T$ as follows: $\min_{q \in Q} |\Lan(\A_{T,q}) \cap \Sigma^{\leq 2n-2}|$. We naturally extend the notion of score to DFAs. Our heuristic algorithm then proceeds as follows:
\begin{enumerate}
	\item Randomly search ($\msf{InitRand}$ attempts) $n$-states transition systems and pick one with the best (i.e., lower) score.
	\item From the current transition system, look through all possible changes of a single transition. Make the change that improves the score the most, if one exists. Otherwise, stop. (Theoretically, this could take exponentially many steps to terminate, in practice, around 30 steps suffice.) 
	\item Do the two above steps $\msf{NbRun}$ times, pick the transition system $T$ with the best score, extract a DFA from $T$ with the best score (over all possible starting states).
\end{enumerate}
This heuristic algorithm has two hyper-parameters: $\msf{InitRand}$ and $\msf{NbRun}$. From several experiments conducted on sample 02 (with $\msf{InitRand} \in \{10,100,1000\}$ and $\msf{NbRun} \in \{10,20,50,100\}$), we have chosen $\msf{InitRand}=100$ and $\msf{NbRun}=50$ as a compromise between running time and efficiency, though it is plausible that better values could be chosen (even for sample 02).

\begin{figure}	
	\begin{minipage}{0.48\linewidth}
		\centering
		\includegraphics[width=1\textwidth]{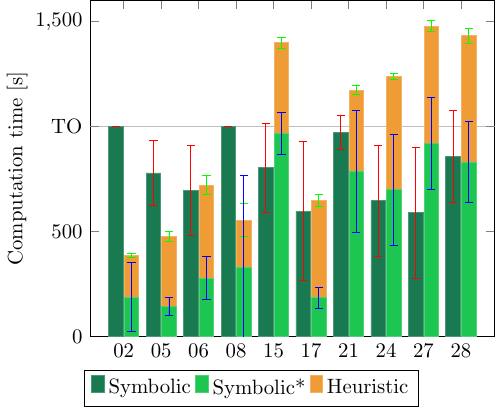}
		\caption{The runtime comparison of the symbolic and the heuristic+symbolic$^*$ algorithms.}	\label{fig:compare_heurisitc}
	\end{minipage}
	\begin{minipage}{0.02\linewidth}
		\phantom{a}
	\end{minipage}
	\begin{minipage}{0.48\linewidth}
		\centering
		\includegraphics[width=1\textwidth]{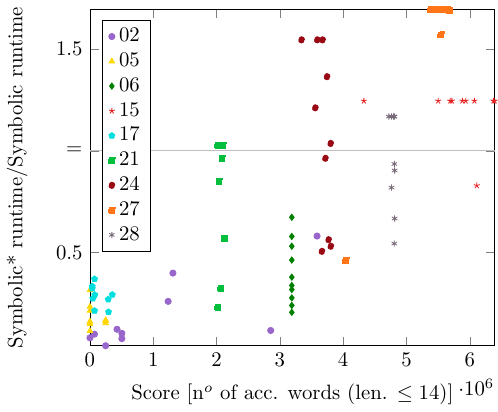}
		\caption{The quotient of the runtime of the symbolic(*) algorithms as a function of the starting score.}
		\label{fig:compare_score}
	\end{minipage}
\end{figure}

\paragraph{Experimental evaluation.} We have implemented this heuristic algorithm in C++ and compared the performance of the symbolic algorithm starting from the trivial accepting DFA and of the symbolic algorithm starting from a DFA computed by the heuristic algorithm, referred to as the symbolic* algorithm. Since the runtime of the heuristic algorithm is not negligible, the experiments are performed on the samples (and number of states $n$) for which the symbolic algorithm takes 
at least 500s. This covers 10 samples, all with $n=8$ except for sample 08, where $n=7$. The results are summarized in Figure~\ref{fig:compare_heurisitc}. 

There are two main positive takeaways from these results. First, there are three samples (02, 05, 08) where our algorithm (heuristic+symbolic*) significantly outperforms the symbolic algorithm. This is a promising result with a portfolio approach in mind where both the heuristic+symbolic* and symbolic algorithms would run in parallel
. Second, there are five samples (02, 05, 06, 08, 17) where the computation time of the symbolic* algorithm is much lower than that of the symbolic algorithm. This result is promising for harder instances where the runtime of the heuristic algorithm becomes more-or-less negligible compared to the runtime of the symbolic algorithm.

On the negative side, on the five other samples, the symbolic* algorithm does not outperform the symbolic algorithm, and even under-performs it for sample 27. This shows that a lower starting score (i.e., score of a starting DFA) does not always mean a better performance of the symbolic algorithm. To better identify how the starting score impacts the runtime of the symbolic algorithm, we have represented relevant data in Figure~\ref{fig:compare_score}. 

In this Figure, each data point corresponds to a run of the symbolic* algorithm\footnote{To compare the scores, we have only depicted the samples where the computation of the heuristic+symbolic* algorithm is done with $n=8$, thus we excluded sample 08.}: in the x-axis, the starting score is depicted, in the y-axis is depicted the quotient of the runtime of the symbolic* algorithm and the mean of the runtime of the symbolic algorithm on that sample. 

From this figure, we can extract the following facts: when starting from a very low score, the runtime ratio is small (as exemplified for samples 02, 05, 17 and also possibly 27); and when starting from higher scores, the impact of the score on the runtime is hard to see. These facts may suggest that the significant differences on the various samples may mostly come from the fact that for some samples the heuristic algorithm has found a starting DFA with a (very) low score, while it is not the case for the other samples. Following, improving the capacity of our heuristic algorithm at finding DFAs with (very) low score (when possible) may have the potential to greatly lower the runtime of the heuristic+symbolic* algorithm. Nonetheless, how the starting score influences the runtime of the SAT solver remains rather cryptic despite our experiments, and further investigating this relationship seems warranted to improve the efficiency of the minimizing-score heuristic approach.

\section{Conclusion}
In this paper, we have studied the task of learning DFAs from positive examples only with a novel word-counting perspective. The core of our studies is theoretical: we have shown the first complexity result on a decision problem directly related to this task, and we have exhibited a new algorithm solving this task with better theoretical guarantees than the state-of-the-art. On the experimental side, our ILP algorithm performs poorly, but our heuristic algorithm shows promising results as a pre-processing tool to use prior to running the symbolic algorithm. Notably, this algorithm shows that transition systems are particularly well-suited for handling this learning from positive examples task.

Our main takeaway is that a better understanding (both theoretically and experimentally) of the relationship between the language minimality of DFAs and their number of accepted words would constitute a valuable insight in the design of efficient algorithms learning DFAs from positive examples only.

\bigskip

\noindent
\textbf{Acknowledgements.} This work has been financially supported by Deutsche Forschungsgemeinschaft, DFG Project number 434592664.

\bibliographystyle{alpha}
\bibliography{ref}

@misc{bordais2025learningdfaspositiveexamples,
	title={Learning DFAs from Positive Examples Only via Word Counting}, 
	author={Benjamin Bordais and Daniel Neider},
	year={2025},
	eprint={2511.08431},
	archivePrefix={arXiv},
	primaryClass={cs.CC},
	url={https://arxiv.org/abs/2511.08431}, 
}

@article{esparza2021automata,
	title={Automata theory, an algorithmic approach. Lecture notes},
	author={Esparza, Javier},
	pages={558},
	year={2021}
}

@inproceedings{DBLP:conf/aaai/0002GBN0T23,
  author       = {Rajarshi Roy and
                  Jean{-}Rapha{\"{e}}l Gaglione and
                  Nasim Baharisangari and
                  Daniel Neider and
                  Zhe Xu and
                  Ufuk Topcu},
  editor       = {Brian Williams and
                  Yiling Chen and
                  Jennifer Neville},
  title        = {Learning Interpretable Temporal Properties from Positive Examples
                  Only},
  booktitle    = {Thirty-Seventh {AAAI} Conference on Artificial Intelligence, {AAAI}
                  2023, Thirty-Fifth Conference on Innovative Applications of Artificial
                  Intelligence, {IAAI} 2023, Thirteenth Symposium on Educational Advances
                  in Artificial Intelligence, {EAAI} 2023, Washington, DC, USA, February
                  7-14, 2023},
  pages        = {6507--6515},
  publisher    = {{AAAI} Press},
  year         = {2023},
  url          = {https://doi.org/10.1609/aaai.v37i5.25800},
  doi          = {10.1609/AAAI.V37I5.25800},
  bibsource    = {dblp computer science bibliography, https://dblp.org}
}

@inproceedings{DBLP:conf/icgi/AvellanedaP18,
  author       = {Florent Avellaneda and
                  Alexandre Petrenko},
  editor       = {Olgierd Unold and
                  Witold Dyrka and
                  Wojciech Wieczorek},
  title        = {Inferring {DFA} without Negative Examples},
  booktitle    = {Proceedings of the 14th International Conference on Grammatical Inference,
                  {ICGI} 2018, Wroc{\l}aw, Poland, September 5-7, 2018},
  series       = {Proceedings of Machine Learning Research},
  volume       = {93},
  pages        = {17--29},
  publisher    = {{PMLR}},
  year         = {2018},
  url          = {http://proceedings.mlr.press/v93/avellaneda19a.html},
  bibsource    = {dblp computer science bibliography, https://dblp.org}
}

@article{zhang2022applications,
  title={Applications of explainable artificial intelligence in diagnosis and surgery},
  author={Zhang, Yiming and Weng, Ying and Lund, Jonathan},
  journal={Diagnostics},
  volume={12},
  number={2},
  pages={237},
  year={2022},
  publisher={MDPI}
}

@article{chen2023explainable,
  title={Explainable artificial intelligence in finance: A bibliometric review},
  author={Chen, Xun-Qi and Ma, Chao-Qun and Ren, Yi-Shuai and Lei, Yu-Tian and Huynh, Ngoc Quang Anh and Narayan, Seema},
  journal={Finance Research Letters},
  volume={56},
  pages={104145},
  year={2023},
  publisher={Elsevier}
}

@article{bharadiya2023artificial,
		title={Artificial intelligence in transportation systems a critical review},
		author={Bharadiya, Jasmin Praful},
		journal={American Journal of Computing and Engineering},
		volume={6},
		number={1},
		pages={35--45},
		year={2023}
	}

@article{weiss2024extracting,
  title={Extracting automata from recurrent neural networks using queries and counterexamples (extended version)},
  author={Weiss, Gail and Goldberg, Yoav and Yahav, Eran},
  journal={Machine Learning},
  volume={113},
  number={5},
  pages={2877--2919},
  year={2024},
  publisher={Springer}
}

@inproceedings{DBLP:conf/aips/CamachoM19,
  author       = {Alberto Camacho and
                  Sheila A. McIlraith},
  editor       = {J. Benton and
                  Nir Lipovetzky and
                  Eva Onaindia and
                  David E. Smith and
                  Siddharth Srivastava},
  title        = {Learning Interpretable Models Expressed in Linear Temporal Logic},
  booktitle    = {Proceedings of the Twenty-Ninth International Conference on Automated
                  Planning and Scheduling, {ICAPS} 2019, Berkeley, CA, USA, July 11-15,
                  2019},
  pages        = {621--630},
  publisher    = {{AAAI} Press},
  year         = {2019},
  url          = {https://ojs.aaai.org/index.php/ICAPS/article/view/3529},
  bibsource    = {dblp computer science bibliography, https://dblp.org}
}

@article{DBLP:journals/ibmrd/RabinS59,
  author       = {Michael O. Rabin and
                  Dana S. Scott},
  title        = {Finite Automata and Their Decision Problems},
  journal      = {{IBM} J. Res. Dev.},
  volume       = {3},
  number       = {2},
  pages        = {114--125},
  year         = {1959},
  url          = {https://doi.org/10.1147/rd.32.0114},
  doi          = {10.1147/RD.32.0114},
  bibsource    = {dblp computer science bibliography, https://dblp.org}
}

@inproceedings{DBLP:conf/aaai/ShvoLIM21,
  author       = {Maayan Shvo and
                  Andrew C. Li and
                  Rodrigo Toro Icarte and
                  Sheila A. McIlraith},
  title        = {Interpretable Sequence Classification via Discrete Optimization},
  booktitle    = {Thirty-Fifth {AAAI} Conference on Artificial Intelligence, {AAAI}
                  2021, Thirty-Third Conference on Innovative Applications of Artificial
                  Intelligence, {IAAI} 2021, The Eleventh Symposium on Educational Advances
                  in Artificial Intelligence, {EAAI} 2021, Virtual Event, February 2-9,
                  2021},
  pages        = {9647--9656},
  publisher    = {{AAAI} Press},
  year         = {2021},
  url          = {https://doi.org/10.1609/aaai.v35i11.17161},
  doi          = {10.1609/AAAI.V35I11.17161},
  bibsource    = {dblp computer science bibliography, https://dblp.org}
}

@article{DBLP:journals/iandc/Gold78,
  author       = {E. Mark Gold},
  title        = {Complexity of Automaton Identification from Given Data},
  journal      = {Inf. Control.},
  volume       = {37},
  number       = {3},
  pages        = {302--320},
  year         = {1978},
  url          = {https://doi.org/10.1016/S0019-9958(78)90562-4},
  doi          = {10.1016/S0019-9958(78)90562-4},
  bibsource    = {dblp computer science bibliography, https://dblp.org}
}

@inproceedings{DBLP:conf/icgi/HeuleV10,
  author       = {Marijn Heule and
                  Sicco Verwer},
  editor       = {Jos{\'{e}} M. Sempere and
                  Pedro Garc{\'{\i}}a},
  title        = {Exact {DFA} Identification Using {SAT} Solvers},
  booktitle    = {Grammatical Inference: Theoretical Results and Applications, 10th
                  International Colloquium, {ICGI} 2010, Valencia, Spain, September
                  13-16, 2010. Proceedings},
  series       = {Lecture Notes in Computer Science},
  volume       = {6339},
  pages        = {66--79},
  publisher    = {Springer},
  year         = {2010},
  url          = {https://doi.org/10.1007/978-3-642-15488-1\_7},
  doi          = {10.1007/978-3-642-15488-1\_7},
  bibsource    = {dblp computer science bibliography, https://dblp.org}
}

@inproceedings{DBLP:conf/cade/GrinchteinLP06,
  author       = {Olga Grinchtein and
                  Martin Leucker and
                  Nir Piterman},
  editor       = {Ulrich Furbach and
                  Natarajan Shankar},
  title        = {Inferring Network Invariants Automatically},
  booktitle    = {Automated Reasoning, Third International Joint Conference, {IJCAR}
                  2006, Seattle, WA, USA, August 17-20, 2006, Proceedings},
  series       = {Lecture Notes in Computer Science},
  volume       = {4130},
  pages        = {483--497},
  publisher    = {Springer},
  year         = {2006},
  url          = {https://doi.org/10.1007/11814771\_40},
  doi          = {10.1007/11814771\_40},
  bibsource    = {dblp computer science bibliography, https://dblp.org}
}

@book{griva2008linear,
  title={Linear and nonlinear optimization 2nd edition},
  author={Griva, Igor and Nash, Stephen G and Sofer, Ariela},
  year={2008},
  publisher={SIAM}
}

@misc{gurobi,
  author = {{Gurobi Optimization, LLC}},
  title = {{Gurobi Optimizer Reference Manual}},
  year = 2024,
  url = "https://www.gurobi.com"
}

@article{DBLP:journals/iandc/Angluin87,
  author       = {Dana Angluin},
  title        = {Learning Regular Sets from Queries and Counterexamples},
  journal      = {Inf. Comput.},
  volume       = {75},
  number       = {2},
  pages        = {87--106},
  year         = {1987},
  url          = {https://doi.org/10.1016/0890-5401(87)90052-6},
  doi          = {10.1016/0890-5401(87)90052-6},
  bibsource    = {dblp computer science bibliography, https://dblp.org}
}

@inproceedings{DBLP:conf/ijcai/BolligHKL09,
  author       = {Benedikt Bollig and
                  Peter Habermehl and
                  Carsten Kern and
                  Martin Leucker},
  editor       = {Craig Boutilier},
  title        = {Angluin-Style Learning of {NFA}},
  booktitle    = {{IJCAI} 2009, Proceedings of the 21st International Joint Conference
                  on Artificial Intelligence, Pasadena, California, USA, July 11-17,
                  2009},
  pages        = {1004--1009},
  year         = {2009},
  url          = {http://ijcai.org/Proceedings/09/Papers/170.pdf},
  bibsource    = {dblp computer science bibliography, https://dblp.org}
}

@inproceedings{DBLP:conf/fm/ShahbazG09,
  author       = {Muzammil Shahbaz and
                  Roland Groz},
  editor       = {Ana Cavalcanti and
                  Dennis Dams},
  title        = {Inferring Mealy Machines},
  booktitle    = {{FM} 2009: Formal Methods, Second World Congress, Eindhoven, The Netherlands,
                  November 2-6, 2009. Proceedings},
  series       = {Lecture Notes in Computer Science},
  volume       = {5850},
  pages        = {207--222},
  publisher    = {Springer},
  year         = {2009},
  url          = {https://doi.org/10.1007/978-3-642-05089-3\_14},
  doi          = {10.1007/978-3-642-05089-3\_14},
  bibsource    = {dblp computer science bibliography, https://dblp.org}
}

@inproceedings{DBLP:conf/ilp/IeloLFRGR23,
	author       = {Antonio Ielo and
		Mark Law and
		Valeria Fionda and
		Francesco Ricca and
		Giuseppe De Giacomo and
		Alessandra Russo},
	editor       = {Elena Bellodi and
		Francesca Alessandra Lisi and
		Riccardo Zese},
	title        = {Towards ILP-Based {LTL} f Passive Learning},
	booktitle    = {Inductive Logic Programming - 32nd International Conference, {ILP}
		2023, Bari, Italy, November 13-15, 2023, Proceedings},
	series       = {Lecture Notes in Computer Science},
	volume       = {14363},
	pages        = {30--45},
	publisher    = {Springer},
	year         = {2023},
	url          = {https://doi.org/10.1007/978-3-031-49299-0\_3},
	doi          = {10.1007/978-3-031-49299-0\_3},
	bibsource    = {dblp computer science bibliography, https://dblp.org}
}

@article{DBLP:journals/jcss/Angluin80,
  author       = {Dana Angluin},
  title        = {Finding Patterns Common to a Set of Strings},
  journal      = {J. Comput. Syst. Sci.},
  volume       = {21},
  number       = {1},
  pages        = {46--62},
  year         = {1980},
  url          = {https://doi.org/10.1016/0022-0000(80)90041-0},
  doi          = {10.1016/0022-0000(80)90041-0},
  bibsource    = {dblp computer science bibliography, https://dblp.org}
}

@article{DBLP:journals/ita/CarrascoO99,
  author       = {Rafael C. Carrasco and
                  Jos{\'{e}} Oncina},
  title        = {Learning deterministic regular grammars from stochastic samples in
                  polynomial time},
  journal      = {{RAIRO} Theor. Informatics Appl.},
  volume       = {33},
  number       = {1},
  pages        = {1--20},
  year         = {1999},
  url          = {https://doi.org/10.1051/ita:1999102},
  doi          = {10.1051/ITA:1999102},
  bibsource    = {dblp computer science bibliography, https://dblp.org}
}

@article{stolcke1992hidden,
  title={Hidden Markov model induction by Bayesian model merging},
  author={Stolcke, Andreas and Omohundro, Stephen},
  journal={Advances in neural information processing systems},
  volume={5},
  year={1992}
}

\appendix
\section{Word-counting and language minimality}
\label{appen:proof_thm}
Here, we establish the folk result Theorem 1. (One can alternatively find a proof and statement in and proof in \cite[Exercice 37]{esparza2021automata}.) The proof of this theorem relies on the lemma* below.
\begin{lemma}
	\label{lem:distinguishing_words}
	Consider a DFA $\A \in \msf{DFA}(\Sigma)$ and let $n := |\A|$. For all pairs of states $(q,q') \in Q^2$, we say that a word $w \in \Sigma^*$ is $(q,q')$-\emph{distinguishing} if $(\delta(q,w),\delta(q',w)) \in F \times (Q \setminus F) \cup (Q \setminus F) \times F$. Then, for all pairs of states $(q,q') \in Q^2$, if there exists a $(q,q')$-distinguishing word, there is one of length at most $n-2$.
\end{lemma}
\begin{proof}
	For all $i \in \N$, we consider the relation $\sim_i \subseteq Q^2$ defined by, for all pairs of states $(q,q')\in Q^2$, $q \sim_i q'$ if and only if there are no $(q,q')$-distinguishing word of length at most $i$. This relation $\sim_i$ is trivially reflexive and symmetric, it is also transitive: if $q \sim_i q' \sim_i q''$, for some $q,q',q'' \in Q$, then for all words $w \in \Sigma^{\leq i}$, we have $\delta(q'',w) \in F \Leftrightarrow \delta(q',w) \in F \Leftrightarrow  \delta(q,w) \in F$. Thus, we have $q \sim_i q''$. In fact, the relation $\sim_i$ is an equivalence relation $Q$. We let $C_i^1,\ldots,C_i^{k_i} \subseteq Q$ denote the sets of $\sim_i$-equivalence classes, with $k_i \geq 1$ denoting the number of $\sim_i$-equivalence classes.
	
	Now, we make three straightforward observations:
	\begin{itemize}
		\item $k_0 = 2$, with $\{ C_0^1,C_0^2 \} = \{ F, Q\setminus F \}$.
		\item For all $i \in \N$, $k_i \leq n$;
		\item For all $i \in \N$, 
		$\sim_{i+1} \; \subseteq \; \sim_i$. 
	\end{itemize}
	
	We have in addition the two, slightly less straightforward, facts:
	\begin{itemize}
		\item For all $i \in \N$, if $k_{i} = k_{i+1}$, then $\sim_{i} = \sim_{i+1}$. Indeed, letting $k := k_i = k_{i+1}$, since $\sim_{i+1} \; \subseteq \; \sim_i$, there is a function $f: \Inter{k} \to \Inter{k}$ such that, for all $j \in \Inter{k}$, we have $C_{i+1}^j \subseteq C_i^{f(j)}$. Since $\cup_{j \in \Inter{k}} C_i^j = Q = \cup_{j \in \Inter{k}} C_{i+1}^j$, the function $f$ is necessarily surjective, and thus bijective, and for all $j \in \Inter{k}$, we have $C_{i+1}^j = C_i^{f(j)}$. That is, the relations $\sim_i$ and $\sim_{i+1}$ have the same equivalence classes, and are thus equal.
		\item For all $i \in \N$, if $\sim_{i} = \sim_{i+1}$, then for all $k \geq i$, we have $\sim_{i} = \sim_k$. Indeed, assuming that $\sim_{i} = \sim_{i+1}$, let us show by induction on $k \geq i$ that $\sim_{k} = \sim_{k+1}$. This directly holds for $k = i$. Assume now that this holds for some $k \geq i$. Let us show that $\sim_{k+1} = \sim_{k+2}$. We have $\sim_{k+2} \subseteq \; \sim_{k+1}$. Furthermore, consider a pair of states $(q,q') \in Q^2$ such that $q \sim_{k+1} q'$. Assume towards a contradiction that there is a $(q,q')$-distinguishing word $w \in \Sigma^{k+2}$. Let us write $w = \alpha \cdot v$, for some $\alpha \in \Sigma$ and $v \in \Sigma^{k+1}$. Hence, the word $v \in \Sigma^{k+1}$ is $(\delta(q,\alpha),\delta(q',\alpha))$-distinguishing. Hence, $\delta(q,\alpha) \not\sim_{k+1} \delta(q',\alpha)$ and $\delta(q,\alpha) \not\sim_{k} \delta(q',\alpha)$. Thus, there is a $(\delta(q,\alpha),\delta(q',\alpha))$-distinguishing $v' \in \Sigma^k$. Then, the word $\alpha \cdot v' \in \Sigma^{k+1}$ is $(q,q')$-distinguishing, thus $q \not \sim_{k+1} q'$. Hence the contradiction. In fact, there is no $(q,q')$-distinguishing word of length $k+2$, therefore $q \sim_{k+2} q'$.
	\end{itemize}
	
	Now, consider the least $m \in \N$ such that $\sim_{m} = \sim_{m+1}$ (which exists since $\sim_0 \supseteq \sim_1 \supseteq \ldots$). We have $2 = k_0 < k_1 < \ldots < k_m \leq n$, thus $m \leq n-2$. (In fact, the relation $\sim_m$ corresponds to the well-known Myhill-Nerode equivalence relation.)
	
	Consider now any pair of states $(q,q') \in Q^2$. Assume that there exists a $(q,q')$-distinguishing word. Therefore, there is some $k \in \N$ such that $q \not\sim_k q'$. Therefore, since $\sim_0 \supseteq \sim_1  \supseteq \ldots \supseteq \sim_m = \sim_{m+1} = \ldots$, we have $\sim_m = \cap_{i \in \N} \sim_i$, and thus $q \not\sim_m q'$. Hence, by definition $\sim_m$, there exists a $(q,q')$-distinguishing word of length at most $m \leq n - 2$. 
\end{proof}

The proof of Theorem 1 is now straightforward.
\begin{proof}
	Let $\A = (Q,\Sigma,\qinit,\delta,F)$ and $\A' = (Q',\Sigma,\qinit',\delta',F')$. Without loos of generality, we may assume that $Q \cap Q' = \emptyset$. Then, let $\A'' = (Q'',\Sigma,\qinit'',\delta'',F'') \in \msf{DFA}(\Sigma)$ be the DFA obtained as the (disjoint) union of $\A$ and $\A'$, that is:
	\begin{itemize}
		\item $Q'' = Q \cup Q'$;
		\item $\qinit'' \in Q''$ is defined arbitrarily;
		\item for all $q \in Q$ (resp. $Q'$) and $\alpha \in \Sigma$, we have $\delta''(q,\alpha) := \delta(q,\alpha) \in Q$ (resp. $\delta''(q,\alpha) := \delta'(q,\alpha) \in Q'$);
		\item $F'' = F \cup F'$.
	\end{itemize}
	Then, we have $|\A''| = |\A| + |\A'|$, and since $\Lan(\A) \neq \Lan(\A')$, then there are $(\qinit,\qinit')$-distinguishing words $u \in \Sigma^*$, of length at most $|\A''| - 2$, by  Lemma*~\ref{lem:distinguishing_words}. By definition, such distinguishing words $u \in \Sigma^{\leq |\A| + |\A'| - 2}$ satisfy $u \in (\Lan(\A) \setminus \Lan(\A')) \cup (\Lan(\A') \setminus \Lan(\A))$.
\end{proof}

\subsection{Two tightness results}
We would like to consider two tightness results: one about Theorem 1, and one about Proposition 1. Specifically, we show the proposition* below. Note that his is a also folk result, see e.g., this \href{https://mathoverflow.net/questions/106905/deciding-equivalence-of-regular-languages}{math overflow thread}.
\begin{proposition}
	Consider a unary alphabet $\Sigma = \{a\}$. For all $m,n \geq 1$, there are two DFAs $\A_m,\A_n \in \msf{DFA}(\Sigma)$ such that $|\A_m| = m$, $|\A_n| = n$, and the length of the smallest word in $(\Lan(\A_n) \setminus \Lan(\A_m)) \cup (\Lan(\A_m) \setminus \Lan(\A_n))$ is $m+n-2$.
\end{proposition}
\begin{proof}
	The result is direct when $m = n = 1$ and $\{m,n\} = \{1,2\}$. Assume now that $m,n \geq 2$. Let us assume without loss of generality that $m \leq n$. We consider the DFAs $\A_m,\A_n$ described below, where the double circled states are accepting:
	
	\vspace*{0.2cm}
	\begin{tikzpicture}
		\node[draw=none] (q) at (0.5,0) {$\A_m:$} ;
		
		\node[scale=1.2,draw,circle,align=center,initial,initial text={}] (q0) at (2,0) {$q_0$} ;
		\node[scale=1.2,draw,circle,align=center] (q1) at (4,0) {$q_1$} ;
		\node[draw=none] (q2) at (6,0) {$\ldots$} ;
		\node[draw,circle,accepting,align=center] (qm2) at (8,0) {$q_{m-2}$} ;
		\node[draw,circle,align=center] (qm1) at (10,0) {$q_{m-1}$} ;
		
		\path[-latex]  				
		(q0) edge node[above] {$a$} (q1)
		(q1) edge node[above] {$a$} (q2)
		(q2) edge node[above] {$a$} (qm2)
		(qm2) edge node[above] {$a$} (qm1)
		(qm1) edge[loop right] node[right] {$a$} (qm1)
		;
	\end{tikzpicture}
	
	\begin{tikzpicture}
		\node[draw=none] (q) at (0.5,0) {$\A_n:$} ;
		
		\node[scale=1.2,draw,circle,align=center,initial,initial text={}] (q0) at (2,0) {$q_0$} ;
		\node[scale=1.2,draw,circle,align=center] (q1) at (4,0) {$q_1$} ;
		\node[draw=none] (q2) at (6,0) {$\ldots$} ;
		\node[draw,circle,accepting,align=center] (qm2) at (8,0) {$q_{m-2}$} ;
		\node[draw=none] (qm) at (10,0) {$\ldots$} ;
		\node[draw,circle,align=center] (qn1) at (12,0) {$q_{n-1}$} ;
		
		\path[-latex]  				
		(q0) edge node[below] {$a$} (q1)
		(q1) edge node[below] {$a$} (q2)
		(q2) edge node[below] {$a$} (qm2)
		(qm2) edge node[below] {$a$} (qm)
		(qm) edge node[below] {$a$} (qn1)
		(qn1) edge[bend right=20] node[above] {$a$} (q0)
		;
	\end{tikzpicture}
	
	Clearly, we have $|\A_m| = m$ and $|\A_n| = n$. Furthermore, we have:
	\begin{equation*}
		\Lan(\A_m) = \{a^{m-2}\} \text{ and } \Lan(\A_n) = \{a^{x} \mid x = m-2 \mod n\}
	\end{equation*}
	Thus, $(\Lan(\A_n) \setminus \Lan(\A_m)) \cup (\Lan(\A_m) \setminus \Lan(\A_n)) = \{a^{x} \mid x \neq m-2,\; x = m-2 \mod n\}$; the smallest word in that set is $a^{m-2+n}$, of length $m+n-2$.
\end{proof}

We also establish the proposition* below.
\begin{proposition}
	Consider a unary alphabet $\Sigma = \{a\}$. For all $n \geq 2$, there is a DFA $\A_n \in \msf{DFA}(\Sigma)$ and a language $L_n \subseteq \Sigma^*$ such that $\A_n \in \msf{Rec}(L_n,n)$, $A_n$ is $\msf{Rec}(L_n,n)$-minimal w.r.t. $\prec_{2n-3}$, but $A_n$ is not $\msf{Rec}(L_n,n)$-minimal w.r.t. $\prec_{\Lan}$.
\end{proposition}
\begin{proof}
	Consider the language $L_n := \{ a^{n-2} \} \subseteq \Sigma^*$ and the DFA $\A_n$ depicted below:
	
	\begin{tikzpicture}
		\node[draw=none] (q) at (0.5,0) {$\A_n:$} ;
		
		\node[scale=1.2,draw,circle,align=center,initial,initial text={}] (q0) at (2.5,0) {$q_0$} ;
		\node[draw=none] (q1) at (5,0) {$\ldots$} ;
		\node[draw,circle,accepting,align=center] (q2) at (7.5,0) {$q_{n-2}$} ;
		\node[draw,circle] (q3) at (10,0) {$q_{n-1}$} ;
		
		\path[-latex]  				
		(q0) edge node[below] {$a$} (q1)
		(q1) edge node[below] {$a$} (q2)
		(q2) edge node[below] {$a$} (q3)
		(q3) edge[bend right=20] node[above] {$a$} (q0)
		;
	\end{tikzpicture}
	
	We have $|\A_n| = n$, $L_n \subseteq \Lan(\A_n) =  \{a^{x} \mid x = n-2 \mod n\}$, and $\Lan(\A_n) \cap \Sigma^{\leq 2n-3} = \{a^{n-2}\}$. Therefore, $A_n$ is $\msf{Rec}(L_n,n)$-minimal w.r.t. $\prec_{2n-3}$. Consider now the DFA $\A_n'$ depicted below:
	
	\begin{tikzpicture}
		\node[draw=none] (q) at (0.5,0) {$\A_n':$} ;
		
		\node[scale=1.2,draw,circle,align=center,initial,initial text={}] (q0) at (2.5,0) {$q_0$} ;
		\node[draw=none] (q1) at (5,0) {$\ldots$} ;
		\node[draw,circle,accepting,align=center] (q2) at (7.5,0) {$q_{n-2}$} ;
		\node[draw,circle] (q3) at (10,0) {$q_{n-1}$} ;
		
		\path[-latex]  				
		(q0) edge node[below] {$a$} (q1)
		(q1) edge node[below] {$a$} (q2)
		(q2) edge node[below] {$a$} (q3)
		(q3) edge[loop right] node[right] {$a$} (q0)
		;
	\end{tikzpicture}
	
	We have $|\A_n'| = n$ and $\Lan(\A_n) = L_n$. Therefore, $A_n$ is not $\msf{Rec}(L_n,n)$-minimal w.r.t. $\prec_{\Lan}$. 
\end{proof}

\section{The NP-hardness proof from Theorem 2}
\label{appen:sec2}
Before we go on to the $\msf{NP}$-hardness proof, let us first consider the (simple) proof of Lemma 1.
\begin{proof}
	For all $k \in \N$, we have:
	\begin{align*}
		\sum_{q' \in Q} \msf{N}(q',k) \cdot |\{ \alpha \in \Sigma \mid \delta(q',\alpha) = q \}| & = \sum_{q' \in Q} |\{ w \in \Sigma^k \mid \delta^*(\qinit,w) = q' \}| \cdot |\{ \alpha \in \Sigma \mid \delta^*(q',w) = q \}| \\
		& = |\{ w \cdot \alpha \in \Sigma^{k+1} \mid \delta^*(\qinit,w) = q \}| = \msf{N}(q,k+1)
	\end{align*}
	Furthermore, we have:
	\begin{align*}
		\Lan(\A) \cap \Sigma^{\leq m} & = \{ w \in \Sigma^{\leq m} \mid \delta^*(q_0,w) \in F \} = \bigcup_{i = 0}^m \{ w \in \Sigma^{i} \mid \delta^*(q_0,w) \in F \} \\
		& = \bigcup_{i = 0}^m \bigcup_{q \in F} \{ w \in \Sigma^{i} \mid \delta^*(q_0,w) = q \} \\
	\end{align*}
	
	Lemma 1 follows.
\end{proof}

Now, let us turn to the $\msf{NP}$-hardness proof. Let us first introduce and recall a few notations. For each $u,v \in \Sigma^*$, we write $u \cdot v \in \Sigma^*$ the concatenation of the word $u$ followed by the word $v$. For all $u \in \Sigma^*$ and $V \subseteq \Sigma^*$, we let $u \cdot V$ and $V \cdot u$ denote the sets $u \cdot V := \{ u \cdot v \mid v \in V \} \subseteq \Sigma^*$ and $V \cdot u := \{ v \cdot u \mid v \in V \} \subseteq \Sigma^*$. A word $u \in \Sigma^*$ is a prefix of a word $v \in \Sigma^*$ if there is some $w \in \Sigma^*$ such that $u \cdot w = v$. This is denoted $u \sqsubseteq v$.

Our goal is to establish the lemma* below.,
\begin{lemma}
	\label{lem:NP-hard}
	The decision problem 1 is $\msf{NP}$-hard, even with $k$ given in unary, and an alphabet $\Sigma$ of size 2.
\end{lemma}


For the remainder of this section, we consider an instance $(X,\mathcal{C})$ of Problem 2 (the All-Pos-Neg SAT problem) with $X = \{ x_i \mid i \in \Intr \}$ the set of variables, and $\mathcal{C} = \{ C_j \mid j \in \Ints \} = \mathcal{C}_+ \uplus \mathcal{C}_-$ the set of clauses. The alphabet that we consider is $\Sigma := \{a,b\}$. 

Our goal is to define a set of positive words $\mathcal{P}$ and three integers $n,k$ and $m$ such that $(X,\mathcal{C})$ is a positive instance of Problem 2 if and only if there is a $(\Pos,n,k,m)$-suitable DFA $\mathcal{A}$, i.e., such that 1) $|\mathcal{A}| \leq n$; 2) $\mathcal{P} \subseteq \Lan(\mathcal{A})$; and 3) $|\Sigma^{\leq m} \cap (\Lan(\mathcal{A}) \setminus \Pos)| \leq k$. Note that this third condition is little different from the one considered in the paper: the maximal length of the words consider is an integer $m$. We will then ensure that our reduction works also for the case $m = 2n-2$.

Let us now formally define the positive set of words $\mathcal{P}$ that we consider. This set of words is parameterized not only by an instance $(X,\mathcal{C})$ of Problem 2, but also by several integers that we do not fix yet. Later (namely, in Lemmas*~\ref{lem:valuation_implies_automaton} and~\ref{lem:assumption_DFA_valuation}), we will show that, assuming that these integers satisfy some (in)equalities, then the set of words that we define below is suitable for a reduction.
As there are many words, these may not be easy to parse. In the next subsection, we will illustrate the DFA construction on an example, which should make it easier to follow how these sets are defined.
\begin{definition}
	\label{def:pos_words_for_reduction}
	Consider an instance $(X,\mathcal{C})$ of the Problem 2. Let $k,d,M,T \in \N$. We define 
	the following sets of words:
	\begin{itemize}
		\item First of all, let us introduce a notation for a set of words: for $\alpha \neq \beta \in \{a,b\}$ and $i \in \N$, we let $U_i(\alpha) := \{ \beta^{x} \cdot \alpha \mid 1 \leq x \leq i+1\}$;
		\item We let
		\begin{equation*}
			\mathcal{P}_{\top}(k,d) := a \cdot a \cdot U_{k+1}(a) \cdot \{ a^{d\phantom{+1}}, a^{d+1} \cdot u_a \mid u_a \in U_{k+1}(a) \}
		\end{equation*}
		and
		\begin{equation*}
			\mathcal{P}_{\bot}(k,d) := a \cdot b \cdot U_{k+1}(a) \cdot \{ a^{d+1}, a^{d+1} \cdot u_a \mid u_a \in U_{k+1}(a) \}
		\end{equation*}
		These are the $\top$- and $\bot$-sets of words. As in the example of the paper, $\top$-words start with $a \cdot a$, and $\bot$-words start with $a \cdot b$.
		\item For all $i \in \Inter{r}$, we let $\mathsf{App}(x_i) := \{ 1 \leq j \leq s \mid x_i \in C_j \}$ and $\mathsf{NotApp}(x_i) := \{ 1 \leq j \leq s \mid x_i \notin C_{j'} \}$ and $\msf{Ind}(x_i) := \mathsf{App}(x_i) \cup \{ s + \sigma \mid \sigma \in \mathsf{NotApp}(x_i)\} \cup \{ s + s + i + t \cdot r \mid t \in \llbracket 0,T-1 \rrbracket \})\}$. Then, we let:		
		\begin{align*}
			\mathcal{P}_{\msf{Var}}(X,\mathcal{C},k,d,T,M,i) := b \cdot a \cdot U_M(b) \cdot b^i \cdot a \cdot b^d \cdot \{ & a^{d+1} \cdot u_a,\; b^{\sigma},\; b^{s + s + T \cdot r} \cdot u_b\\
			\mid & u_a \in U_{k+1}(a),\; \sigma \in \msf{Ind}(x_i),\; u_b \in U_{k+1}(b) \}
		\end{align*}
		These are the variable words, which all start with $b \cdot a$. 
		\item For all $j \in \Inter{s}$, we let:
		\begin{equation*}
			\mathcal{P}_{\msf{Cl}}(X,\mathcal{C},k,d,T,j) := b \cdot b \cdot a^j \cdot b \cdot b^d \cdot \{ b^{j},b^{s + s + T \cdot r} \cdot u_b \mid u_b \in U_{k+1}(b)\}
		\end{equation*}
		and
		\begin{align*}
			\mathcal{P}^{\mathsf{acc}}_{\msf{Cl}}(k,d,j) :=
			b \cdot b \cdot a^j \cdot b \cdot b^d \cdot  \begin{cases}
				\{ a^{d\phantom{+1}},a^{d+1} \cdot u_a \mid u_a \in U_{k+1}(a) \} & \text{ if }C_j \in \mathcal{C}_+ \\
				\{ a^{d+1},a^{d+1} \cdot u_a \mid u_a \in U_{k+1}(a) \} & \text{ if }C_j \in \mathcal{C}_- \\
			\end{cases}
		\end{align*}
		These are the clause words, which all start with $b \cdot b$. Note that clause words in $\mathcal{P}^{\mathsf{acc}}_{\msf{Cl}}(k,d,j)$ correspond to those defined in the paper that encode that some clause is positive while some other clause is negative.
	\end{itemize}  
	We can now define the set of words that we consider:
	\begin{align*}
		\mathcal{P}(X,\mathcal{C},k,d,T,M) & := \left(\bigcup_{i \in \Inter{r}}\mathcal{P}_{\msf{Var}}(X,\mathcal{C},k,d,T,M,i)\right) \cup \left(\bigcup_{j \in \Inter{s}}\mathcal{P}_{\msf{Cl}}(X,\mathcal{C},k,d,T,j) \cup \mathcal{P}^{\mathsf{acc}}_{\msf{Cl}}(k,d,j)\right)\\
		& \phantom{a}\cup \mathcal{P}_{\top}(k,d) \cup \mathcal{P}_{\bot}(k,d) 
	\end{align*} 
	
	We also define two integers that will be particularly useful when discussing the size of DFAs:
	\begin{equation*}
		\omega_1(X,d) := d \cdot (2 + r)
	\end{equation*}
	and 
	\begin{equation*}
		\omega_2(X,\mathcal{C},k,T,M) := 18 + s + M + 4 \cdot k + r + 2 \cdot r \cdot s + r^2 \cdot T
	\end{equation*}
\end{definition}


\subsection{If there is a satisfying valuation, then there is a suitable DFA}
Let us first show that, assuming that $n$ and $k$ are large enough, if there exists a valuation satisfying $(X,\mathcal{C})$, then there is a DFA satisfying conditions a), b), c). This is stated in the lemma* below.
\begin{lemma}
	\label{lem:valuation_implies_automaton}
	Consider a positive instance $(X,\mathcal{C})$ of Problem 2. Let $n,k,d,M,T \in \N$ and assume that $n \geq \omega_1(r,d) + \omega_2(X,\mathcal{C},k,T,M)$ and $k \geq M \cdot r + s \cdot (s + T -1)$. Then, for all $m \in \N$, there is a $(\Pos,n,k,m)$-suitable DFA.
\end{lemma} 

To prove this lemma* we define a DFA parameterized by the above integers and a satisfying valuation of the variables. This definition is illustrated on an example afterwards. We invite the reader to glimpse at Figure~\ref{fig:test1} to get an idea on the shape of the DFA that we define.
\begin{definition}
	\label{def:DFA_from_input}
	Consider an instance $(X,\mathcal{C})$ of Problem 2 and a valuation $\nu: X \to \{\top,\bot\}$ satisfying $(X,\mathcal{C})$. We define the DFA $\mathcal{A}(X,\mathcal{C},k,d,T,M,\nu) = (Q,\Sigma,\qinit,\delta,F)$ where: 
	\begin{equation*}
		Q := Q_1 \uplus Q_{\msf{Cl}} \uplus Q_{\msf{Var}} \uplus Q_{\top} \uplus Q_{\bot} \uplus Q_{u_a} \uplus Q_{u_b}
	\end{equation*}
	with 
	\begin{align*}
		Q_1 & := \{ \qinit,q_a,q_b,q_{\msf{acc}},q_{\msf{
				rej}} \} \\
		Q_{\msf{Cl}} & := \{ q_{\msf{Cl}},Q_{C_j} \mid j \in \Inter{s} \} \\
		Q_{\msf{Var}} & := \{ q_{\msf{Var}},q_{\msf{Var},u}^\sigma,q_X \mid \sigma \in \Inter{M} \} \uplus \cup_{i \in \Inter{r}} Q_{x_i} \\
		Q_{x_i} & := \{ q_{x_i},q_{x_i}^\sigma,(x_i \in C_j),(x_i \notin C_j),(x_i,l,t) \mid \sigma \in \Inter{d}, j \in \Inter{s},l \in \Inter{r}, t \in \llbracket 0,T-1 \rrbracket\}\\
		Q_\top & := \{ q_\top,q_{\top,u}^\sigma,q_{\top}^{\sigma'} \mid \sigma \in \Inter{k+1},\sigma' \in \llbracket 0,d+1 \rrbracket \} \\
		Q_\bot & := \{ q_\bot,q_{\bot,u}^{\sigma},q_{\bot}^{\sigma'} \mid \sigma \in \Inter{k+1},\sigma' \in \llbracket 0,d+1 \rrbracket \} \\
		Q_{u_a} & := \{ q_{u_a}^\sigma \mid \sigma \in \Inter{k+1} \} \\
		Q_{u_b} & := \{ q_{u_b}^\sigma \mid \sigma \in \Inter{k+1} \}
	\end{align*}
	Before we finish the definition of this DFA, let us count its number of states. We have:
	\begin{align*}
		|Q| 
		& = \underbrace{5}_{|Q_1|} + \underbrace{s + 1}_{|Q_{\msf{Cl}}|} + \underbrace{2 + M + r \cdot (1 + d + 2 \cdot s+r \cdot T)}_{|Q_{\msf{Var}}|} + \underbrace{d + k + 4}_{|Q_{\top}|} + \underbrace{d + k + 4}_{|Q_{\bot}|} + \underbrace{k + 1}_{|Q_{u_a}|} + \underbrace{k + 1}_{|Q_{u_b}|} \\
		& = d \cdot (2 + r) + 18 + s + M + 4 \cdot k + r + 2 \cdot r \cdot s + r^2 \cdot T \\
		& = \omega_1(X,d) +	\omega_2(X,\mathcal{C},k,T,M)
	\end{align*}
	
	The set of final states $F \subseteq Q$ is equal to:
	\begin{equation*}
		F := \{ q_\top^d,q_\bot^{d+1},q_{\msf{acc}}, \} \uplus \bigcup_{i \in \Inter{r}} \left(\bigcup_{j \in \mathsf{App}(x_i)} (x_i \in C_j) \cup \bigcup_{j \in \mathsf{NotApp}(x_i)} (x_i \notin C_j) \cup \bigcup_{t \in \llbracket 0,T-1 \rrbracket} (x_i,i,t)\right)
	\end{equation*}
	
	Let us now define the transition function $\delta$. We start by defining those transitions that do not depend on the valuation $\nu: X \to \{\top,\bot\}$. First, we define the initial transitions leading from $\qinit$ to $q_\top,q_\bot,q_{\msf{Var}},q_{\msf{Cl}}$. Specifically, we have: 
	\begin{equation*}
		\delta(\qinit,a) := q_a, \; \delta(\qinit,b) := q_b, \; \delta(q_a,a) := q_\top, \; \delta(q_a,b) := q_\bot, \;  \delta(q_b,a) := q_{\msf{Var}}, \; \delta(q_b,b) := q_{\msf{Cl}}
	\end{equation*}
	Let us now define the other transitions (that still do not depend on the valuation $\nu$). To improve the readability of the definition of these other transitions, we describe them in a somewhat graphical manner. Specifically:
	\begin{itemize}
		\item for all $q,q' \in Q$ and $\alpha \in \Sigma$, we write $q \xlongrightarrow{\alpha} q'$ to express that $\delta(q,\alpha) := q'$;
		\item for all $q_0,q_1,\ldots,q_\theta \in Q$ and $\alpha \neq \beta \in \Sigma$, we write $q_0 \xlongrightarrow{\alpha} q_1 \xlongrightarrow{\alpha} \ldots \xlongrightarrow{\alpha} q_\theta \xRightarrow{\beta} q$ to express that, for all $0 \leq \sigma \leq \theta-1$, we have $\delta(q_\sigma,\alpha) := q_{\sigma+1}$ and $\delta(q_\sigma,\beta) := q$, while $\delta(q_\theta,\beta) := q$ (but we do not have $\delta(q_\theta,\alpha) := q$).
	\end{itemize}
	Furthermore, for almost all transitions $q \xlongrightarrow {\alpha}q'$ that we describe below, there is some $\bullet \in \{\msf{Cl},\msf{Var},\top,\bot,u_a,u_b\}$ such that $q,q' \in Q_{\bullet}$. We therefore group these transitions by the set of states $Q_\bullet$ to which all states, but the ones we depict in boxes, belong.
	
	\begin{align*}
		Q_{\msf{Cl}}:\; & q_{\msf{Cl}} \xlongrightarrow[]{a} q_{C_1} \xlongrightarrow[]{a} q_{C_2} \xlongrightarrow[]{a} \ldots \xlongrightarrow[]{a} q_{C_s} \\
		Q_{\msf{Var}}:\; & q_{\msf{Var}} \xlongrightarrow[]{a} q_{\msf{Var},u}^1 \\
		& q_{\msf{Var},u}^1 \xlongrightarrow[]{a} q_{\msf{Var},u}^2 \xlongrightarrow[]{a} \ldots \xlongrightarrow[]{a} q_{\msf{Var},u}^M \xRightarrow{b} q_X \xlongrightarrow{b} q_{x_1} \xlongrightarrow{b} q_{x_2} \xlongrightarrow{b} \ldots \xlongrightarrow{b} q_{x_r} \\
		& \forall i \in \Inter{r}:\; q_{x_i} \xlongrightarrow{a} q_{x_i^0} \xlongrightarrow{b} q_{x_1^1} \xlongrightarrow{b} \ldots \xlongrightarrow{b} q_{x_i^d} \xlongrightarrow{b} \\
		& \phantom{\forall i \in \Inter{r}:\;}(x_i \in C_1) \xlongrightarrow{b} (x_i \in C_2) \xlongrightarrow{b} \ldots \xlongrightarrow{b} (x_i \in C_s) \xlongrightarrow{b} \\
		& \phantom{\forall i \in \Inter{r}:\;}(x_i \notin C_1) \xlongrightarrow{b} (x_i \notin C_2) \xlongrightarrow{b} \ldots \xlongrightarrow{b} (x_i \notin C_s)\xlongrightarrow{b}\\
		& \phantom{\forall i \in \Inter{r}:\;}(x_i,1,0) \xlongrightarrow{b} (x_i,2,0) \xlongrightarrow{b} \ldots \xlongrightarrow{b} (x_i,r,0) \xlongrightarrow{b} \\
		& \phantom{\forall i \in \Inter{r}:\;}(x_i,1,1) \xlongrightarrow{b} (x_i,2,1) \xlongrightarrow{b} \ldots
		\xlongrightarrow{b} (x_i,r,1) 
		\xlongrightarrow{b} \\
		& \phantom{\forall i \in \Inter{r}:\;}\ldots 
		\xlongrightarrow{b} \\ 
		& \phantom{\forall i \in \Inter{r}:\;}(x_i,1,T-1) \xlongrightarrow{b} (x_i,2,T-1) \xlongrightarrow{b} \ldots
		\xlongrightarrow{b} (x_i,r,T-1) \xlongrightarrow[]{a} \boxed{q_{u_b}^1}\\
		Q_{\top}:\; & q_{\top} \xlongrightarrow[]{a} q_{\top,u}^1\\
		&q_{\top,u}^1 \xlongrightarrow[]{b} q_{\top,u}^2 \xlongrightarrow[]{b} \ldots \xlongrightarrow[]{b} q_{\top,u}^{k+1} \xRightarrow{a} q_\top^1 \xlongrightarrow[]{a} q_\top^2 \xlongrightarrow[]{a} \ldots \xlongrightarrow[]{a} q_\top^{d+1} \xlongrightarrow[]{a} \boxed{q_{u_a}^1} \\
		Q_{\bot}:\; & q_{\bot} \xlongrightarrow[]{a} q_{\bot,u}^1\\
		&q_{\bot,u}^1 \xlongrightarrow[]{b} q_{\bot,u}^2 \xlongrightarrow[]{b} \ldots \xlongrightarrow[]{b} q_{\bot,u}^{k+1} \xRightarrow{a} q_\bot^1 \xlongrightarrow[]{a} q_\bot^2 \xlongrightarrow[]{a} \ldots \xlongrightarrow[]{a} q_\bot^{d+1} \xlongrightarrow[]{a} \boxed{q_{u_a}^1} \\
		Q_{u_a}:\; & q_{u_a}^1 \xlongrightarrow[]{b} q_{u_a}^2 \xlongrightarrow[]{b} \ldots \xlongrightarrow[]{b} q_{u_a}^{k+1} \xRightarrow{a} \boxed{q_{\msf{acc}}}\\
		Q_{u_b}:\; & q_{u_b}^1 \xlongrightarrow[]{a} q_{u_b}^2 \xlongrightarrow[]{a} \ldots \xlongrightarrow[]{a} q_{u_b}^{k+1} \xRightarrow{b} \boxed{q_{\msf{acc}}}
	\end{align*}
	Let us now describe the transitions that depend on the valuation $\nu$. Since $\nu$ satisfies $(X,\mathcal{C})$, for all $j \in \Inter{s}$, there is some $i_j \in \Inter{r}$ such that $x_{i_j} \in C_j$ and, if $C_j \in \mathcal{C}_+$ (resp. $C_j \in \mathcal{C}_-$), then $v(x_{i_j}) = \top$ (resp. $v(x_{i_j}) = \bot$). Then, we have transitions from states in $Q_{\msf{Cl}}$ to states in $Q_\msf{Var}$, and from states in $Q_{\msf{Var}}$ to states in $Q_\top \cup Q_\bot$. Specifically:
	\begin{itemize}
		\item From $Q_{\msf{Cl}}$ to $Q_\msf{Var}$: for all $j \in \Inter{s}$, $\delta(q_{C_j},b) := q_{x_{i_j}}^0$.
		\item From $Q_{\msf{Var}}$ to $Q_\top \cup Q_\bot$: for all $i \in \Inter{r}$, $\delta(q_{x_i}^d,a) := q_{\nu(x_i)}^1$.
	\end{itemize}
	
	Finally, for all $q \in Q$ and $\alpha \in \Sigma$, if $\delta(q,\alpha)$ is not defined above, then we have $\delta(q,\alpha) := q_{\msf{rej}}$. In particular, this implies that $\delta(q_{\msf{rej}},a) = \delta(q_{\msf{rej}},b) = \delta(q_{\msf{acc}},a) = \delta(q_{\msf{acc}},b) := q_{\msf{rej}}$, thus $q_{\msf{rej}}$ is a self-looping sink state from which all words are rejected.
\end{definition}

We illustrate the above definition in Example*~\ref{ex:DFA}. 
\begin{figure}
	\centering
	\begin{minipage}{.6\textwidth}
		\hspace*{-1cm}
		\centering
		\includegraphics[
		height=23cm]{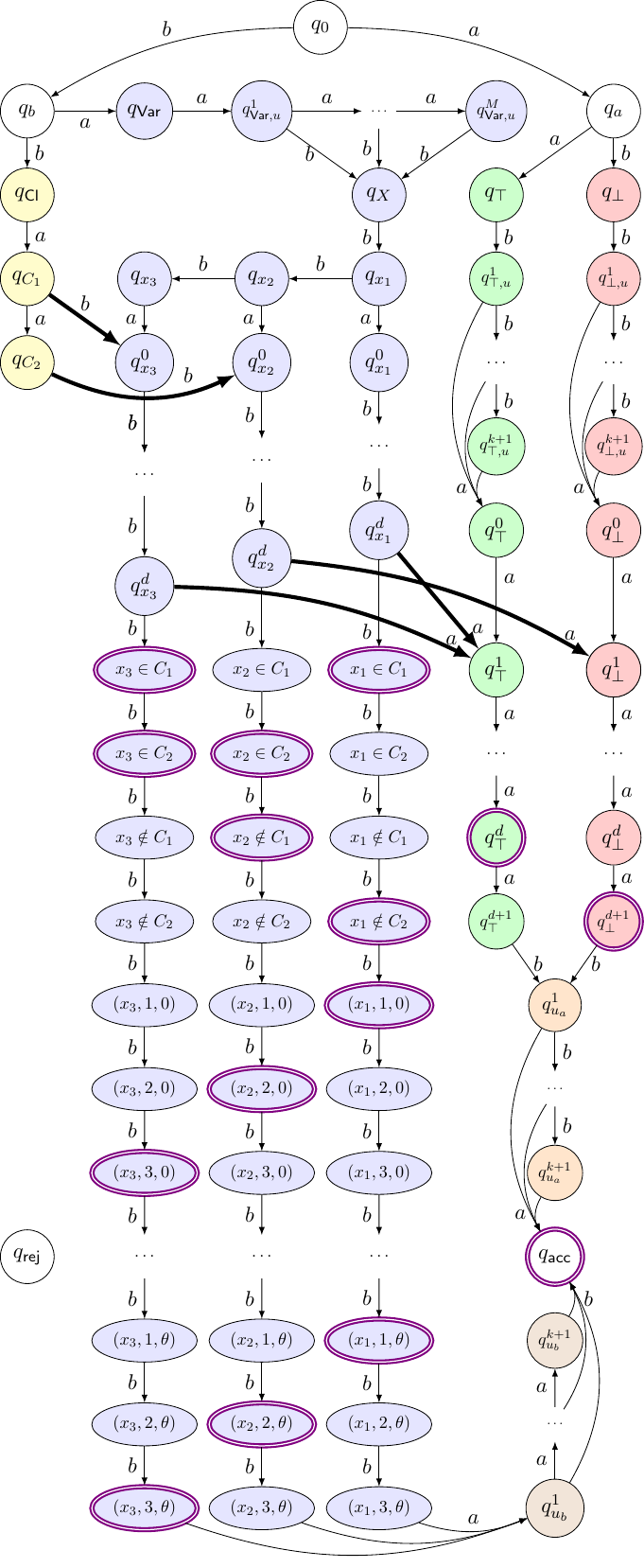}
		\caption{A DFA $\mathcal{A}(X,\mathcal{C},k,d,T,M,\nu)$.}
		\label{fig:test1}
	\end{minipage}%
	\begin{minipage}{.4\textwidth}
		\hspace*{-5cm}
		\begin{example}
			\label{ex:DFA}
			Consider the set of variables $X = \{ x_1,x_2,x_3 \}$, and a positive clause $C_1 = \{ x_1,x_3 \} \in \mathcal{C}_+$ and a negative clause $C_2 = \{ x_2,x_3 \} \in \mathcal{C}_-$. Consider the valuation $\nu: X \to \{\top,\bot\}$ such that $\nu(x_1) = \nu(x_3) := \top$, and $\nu(x_2) := \bot$. This valuation satisfies $(X,\mathcal{C})$ with $i_1 := 3$ (since $x_3 \in C_1$, $C_1 \in \mathcal{C}_+$, and $\nu(x_3) = \top$), and $i_2 := 2$ (since $x_2 \in C_2$, $C_2 \in \mathcal{C}_-$, and $\nu(x_2) = \bot$). Given some integers $k,d,T,M$, the DFA $\mathcal{A} := \mathcal{A}(X,\mathcal{C},k,d,T,M,\nu)$ of Definition~\ref{def:DFA_from_input} is depicted in Figure~\ref{fig:test1} (with $\theta := T-1$). We let $\mathcal{P} := \mathcal{P}(X,\mathcal{C},k,d,T,M)$.
			
			In this DFA, states are colored in white, yellow, blue, green, red, orange, and brown according to the set of states in $\{Q_1,Q_{\msf{Cl}},Q_{\msf{Var}},Q_{\top},Q_{\bot},Q_{u_a},Q_{u_b}\}$ to which they belong. We did not depict any transition 
			to the self-looping sink state $q_{\msf{rej}}$.
			
			Accepting states are depicted with a double-circled violet border. For instance: $(x_1 \in C_1) \in Q$ is accepting (and thus $(x_1 \notin C_1) \in Q$ is not accepting) since $x_1 \in C_1$, while $(x_2 \in C_1) \in Q$ is not accepting (and thus $(x_2 \notin C_1) \in Q$ is accepting) since $x_2 \notin C_2$. We also have $q^d_\top$ accepting and $q^{d+1}_\top$ not accepting, while $q^d_\bot$ is not accepting and $q^{d+1}_\bot$ is accepting.
			
			The transitions that depend on the valuation $\nu$ are depicted in bold, with the transitions from $Q_{\msf{Cl}}$ to $Q_{\msf{Var}}$ defined by the indices $i_j \in \Inter{3}$, for $j \in \Inter{2}$, while with the transitions from $Q_{\msf{Var}}$ to $Q_{\top} \cup Q_{\bot}$ are given directly by the valuation $\nu$. Since this valuation $\nu$ satisfies $(X,\mathcal{C})$, we have that:
			\begin{itemize}
				\item the word $b \cdot b \cdot a^1 \cdot b \cdot b^d \cdot a^d \in \mathcal{P}$ (since $C_1 \in \mathcal{C}_+$) is accepted by $\mathcal{A}$ since $\delta(q_{C_1},b) = q_{x_3}^0$, and $\delta(q_{x_3}^d,a) = q_{\top}^1$;
				\item the word $b \cdot b \cdot a^2 \cdot b \cdot b^d \cdot a^{d+1} \in \mathcal{P}$ (since $C_2 \in \mathcal{C}_-$) is accepted by the DFA $\mathcal{A}$ since $\delta(q_{C_2},b) = q_{x_2}^0$, and $\delta(q_{x_2}^d,a) = q_{\bot}^1$.
			\end{itemize}
		\end{example}
	\end{minipage}
\end{figure}

Let us now proceed to the proof of Lemma*~\ref{lem:valuation_implies_automaton}. This proof is not hard, as it essentially consists in computing the language accepted by the DFA of Definition~\ref{def:DFA_from_input}, however it is quite long since we compute this language step by step.
\begin{proof}
	Consider a positive instance $(X,\mathcal{C})$ of Problem 2 and let $n,k,d,M,T \in \N$ be such that $n \geq \omega_1(r,d) + \omega_2(X,\mathcal{C},k,T,M)$ and $k \geq M \cdot r + s \cdot (s + T -1)$. Let $m \in \N$. 
	
	By definition, there is some valuation $\nu: X \to \{\top,\bot\}$ satisfying $(X,\mathcal{C})$. Consider the DFA $\mathcal{A} := \mathcal{A}(X,\mathcal{C},k,d,T,M,\nu) = (Q,\Sigma,\qinit,\delta,F)$ from Definition~\ref{def:DFA_from_input}. We have $|Q| \leq n$ by assumption. Let us now compute the language $\Lan(\mathcal{A}) \subseteq \Sigma^*$ and relate it to the set of positive words $\mathcal{P}(X,\mathcal{C},k,d,T,M)$ from Definition~\ref{def:pos_words_for_reduction}. To compute this language, we will, as intermediary steps, compute sets of words accepted from various states $q \in Q$, i.e. the language:
	\begin{equation*}
		\Lan(\mathcal{A},q) := \{ w \in \Sigma^* \mid \delta^*(q,w) \in F \}
	\end{equation*} 
	Let us start by computing the language $\Lan(\mathcal{A},q^1_{u_a})$ and $\Lan(\mathcal{A},q^1_{u_b})$. The state $q_{\msf{acc}}$ is the single accepting state reachable from the state $q^1_{u_a}$. From the definition of the transition function $\delta$ for states in $Q_{u_a}$, it is clear that we have:
	\begin{align*}
		& \Lan(\mathcal{A},q^{k+1}_{u_a}) = \{ a \} \\
		\forall \sigma \in \llbracket 2,k+1 \rrbracket,\; & \Lan(\mathcal{A},q^{\sigma-1}_{u_a}) = \{ a \} \cup b \cdot \Lan(\mathcal{A},q^{\sigma}_{u_a})
	\end{align*} 
	Hence, we have:
	\begin{equation*}
		\Lan(\mathcal{A},q^{1}_{u_a}) = \{ b^\sigma \cdot a \mid \sigma \in \llbracket 0,k \rrbracket \} \text{ and } b \cdot \Lan(\mathcal{A},q^{1}_{u_a}) = U_{k+1}(a)
	\end{equation*}
	Similarly, we have:
	\begin{equation*}
		\Lan(\mathcal{A},q^{1}_{u_b}) = \{ a^\sigma \cdot b \mid \sigma \in \llbracket 0,k \rrbracket \} \text{ and } a \cdot \Lan(\mathcal{A},q^{1}_{u_b}) = U_{k+1}(b)
	\end{equation*}
	
	In a similar manner, letting $\Lan(\mathcal{A},q,q') := \{ w \in \Sigma^* \mid \delta^*(q,w) = q' \}$ for all $q,q' \in Q$, we have $\Lan(\mathcal{A},q_{\top},q_\top^0) = \Lan(\mathcal{A},q_{\bot},q_\bot^0) = U_{k+1}(a)$, and $\Lan(\mathcal{A},q_{\msf{Var}},q_X) = U_M(b)$.
	
	Let us now compute the languages $\Lan(\mathcal{A},q^1_\top)$ and $\Lan(\mathcal{A},q^1_\bot)$. A word in $\Sigma^*$ is accepted from the state $q^1_\top$ if, upon reading it from $q^1_\top$, we either stop at the accepting state $q^d_\top$, or we visit the state $q^{1}_{u_a}$ and the corresponding suffix is in the language $\Lan(\mathcal{A},q^{1}_{u_a})$. Therefore, we have:
	\begin{equation*}
		\Lan(\mathcal{A},q^{1}_\top) = \{a^{d-1}\} \cup a^{d} \cdot b \cdot \Lan(\mathcal{A},q^{1}_{u_a}) = \{a^{d-1}\} \cup a^{d} \cdot U_{k+1}(a)
	\end{equation*}
	Similarly, we have:
	\begin{equation*}
		\Lan(\mathcal{A},q^{1}_\bot) = \{a^{d}\} \cup a^{d} \cdot b \cdot \Lan(\mathcal{A},q^{1}_{u_a}) = \{a^{d}\} \cup a^{d} \cdot U_{k+1}(a)
	\end{equation*}
	
	Now, let us compute the languages $\Lan(\mathcal{A},q_\top)$ and $\Lan(\mathcal{A},q_\bot)$. We start with $\Lan(\mathcal{A},q_\top)$:
	\begin{equation*}
		\Lan(\mathcal{A},q_\top) = \Lan(\mathcal{A},q_\top,q_\top^0) \cdot a \cdot \Lan(\mathcal{A},q^{1}_\top) = U_{k+1}(a) \cdot \left(\{a^{d}\} \cup a^{d+1} \cdot U_{k+1}(a)\right)
	\end{equation*}
	Hence, we have:
	\begin{equation}
		\label{eqn:aa}
		\Lan(\mathcal{A}) \cap a \cdot a \cdot \Sigma^* = a \cdot a \cdot \Lan(\mathcal{A},q_\top) = a \cdot a \cdot U_{k+1}(a) \cdot \left(\{a^{d}\} \cup a^{d+1} \cdot U_{k+1}(a)\right) = \mathcal{P}_\top(k,d)
	\end{equation}
	Similarly, we have:
	\begin{equation*}
		\Lan(\mathcal{A},q_\bot) = U_{k+1}(a) \cdot \left(\{a^{d+1}\} \cup a^{d+1} \cdot U_{k+1}(a)\right)
	\end{equation*}
	and 
	\begin{equation}
		\label{eqn:ab}
		\Lan(\mathcal{A}) \cap a \cdot b \cdot \Sigma^* = a \cdot b \cdot U_{k+1}(a) \cdot \left(\{a^{d+1}\} \cup a^{d+1} \cdot U_{k+1}(a)\right) = \mathcal{P}_\bot(k,d)
	\end{equation}
	
	Let us now compute the language $\Lan(\mathcal{A},q_{\msf{Var}})$. To do so, we let $i \in \Inter{r}$ and we compute the language $b \cdot \Lan(\mathcal{A},(x_i \in C_1))$. We have:
	\begin{align*}
		b \cdot \Lan(\mathcal{A},(x_i \in C_1)) & = \{ b^\sigma \mid \sigma \in \msf{Ind}(x_i) \} \cup b^{s + s + T \cdot r} \cdot a \cdot \Lan(\mathcal{A},q^1_{u_b})
	\end{align*}
	with $a \cdot \Lan(\mathcal{A},q^1_{u_b}) = U_{k+1}(b)$ (as established above). Furthermore, we have:
	\begin{equation*}
		\Lan(\mathcal{A},q_{x_i}^0) = b^d \cdot \left(a \cdot \Lan(\mathcal{A},q^1_{\nu(x_i)}) \cup b \cdot \Lan(\mathcal{A},(x_i \in C_1)) \right)
	\end{equation*}
	and
	\begin{equation*}
		\Lan(\mathcal{A},q_{\msf{Var}}) = \Lan(\mathcal{A},q_{\msf{Var}},q_X) \cdot \left(\bigcup_{i \in \Inter{r}} \Lan(\mathcal{A},q_{X},q_{x_i}^0) \cdot \Lan(\mathcal{A},q_{x_i}^0)\right) = U_{M}(b) \cdot \left(\bigcup_{i \in \Inter{r}} b^i \cdot a \cdot \Lan(\mathcal{A},q_{x_i}^0)\right)
	\end{equation*}
	Therefore, letting $d(\top) := d$ and $d(\bot) := d+1$, and $L := \Lan(\mathcal{A}) \cap b \cdot a \cdot \Sigma^*$, we have:
	\begin{align*}
		L = & \; b \cdot a \cdot \Lan(\mathcal{A},q_{\msf{Var}}) = b \cdot a \cdot U_{M}(b) \cdot \left(\bigcup_{i \in \Inter{r}} b^i \cdot a \cdot \Lan(\mathcal{A},q_{x_i}^0) \right) \\
		= & \bigcup_{i \in \Inter{r}} b \cdot a \cdot U_{M}(b) \cdot b^i \cdot a \cdot \Lan(\mathcal{A},q_{x_i}^0) \\
		= & \bigcup_{i \in \Inter{r}} b \cdot a \cdot U_{M}(b) \cdot b^i \cdot a \cdot b^{d} \cdot \left(\underbrace{a \cdot \Lan(\mathcal{A},q^1_{\nu(x_i)})}_{= \{a^{d(\nu(x_i))}\} \; \cup \; a^{d+1} \cdot U_{k+1}(a)} \cup \{ b^\sigma \mid \sigma \in \msf{Ind}(x_i) \}  \cup d^{s + s + T \cdot r} \cdot U_{k+1}(b) \right) \\
		= & \bigcup_{i \in \Inter{r}} b \cdot a \cdot U_M(b) \cdot b^i \cdot a \cdot b^{d} \cdot \left(a^{d+1} \cdot U_{k+1}(a) \cup \{ b^\sigma \mid \sigma \in \msf{Ind}(x_i) \} \cup d^{s + s + T \cdot r} \cdot U_{k+1}(b)\right) \\
		\cup & \bigcup_{i \in \Inter{r}} b \cdot a \cdot U_M(b) \cdot b^i \cdot a \cdot b^{d} \cdot a^{d(\nu(x_i))}\\\
	\end{align*}
	Therefore, letting $\msf{ErrVar} := \bigcup_{i \in \Inter{r}} b \cdot a \cdot U_M(b) \cdot b^i \cdot a \cdot b^{d} \cdot a^{d(\nu(x_i))}$, we obtain:
	\begin{equation}
		\label{eqn:ba}
		\Lan(\mathcal{A}) \cap b \cdot a \cdot \Sigma^* = \bigcup_{i \in \Inter{r}} \mathcal{P}_{\msf{Var}}(X,\mathcal{C},k,d,T,M,i) \cup \msf{ErrVar}
	\end{equation}
	
	Let us now finally turn to the language $\Lan(\mathcal{A},q_{\msf{Cl}})$. Clearly, we have:
	\begin{align*}
		\Lan(\mathcal{A},q_{\msf{Cl}}) & = \bigcup_{j \in \Inter{s}} \Lan(\mathcal{A},q_{\msf{Cl}},q_{C_j}) \cdot b \cdot \Lan(\mathcal{A},q^0_{x_{i_j}}) = \bigcup_{j \in \Inter{s}} a^j \cdot b \cdot \Lan(\mathcal{A},q^0_{x_{i_j}}) \\
		& = \bigcup_{j \in \Inter{s}} a^j \cdot b \cdot \left(b^d \cdot a \cdot \Lan(\mathcal{A},q^1_{\nu(x_{i_j})}) \cup b^d \cdot b \cdot \Lan(\mathcal{A},(x_{i_j} \in C_1))\right) \\
		& = \bigcup_{j \in \Inter{s}} a^j \cdot b \cdot \left( \{b^d \cdot a^{d(\nu(x_{i_j}))}\} \; \cup \; b^d \cdot  a^{d+1} \cdot U_{k+1}(a) \right.\\
		& \phantom{bigcup_{j \in \Inter{s}} a^j }
		\left.\cup \; b^d \cdot \{ b^{\sigma} \mid \sigma \in \msf{Ind}(x_{i_j}) \} \cup b^{s + s + T \cdot r} \cdot U_{k+1}(b)\right)
	\end{align*}
	For all $j \in \Inter{s}$, since $x_{i_j} \in C_j$, we have $j \in \msf{App}(x_{i_j}) \subseteq \msf{Ind}(x_{i_j})$. Hence, we have:
	\begin{align*}
		b \cdot b \cdot \Lan(\mathcal{A},q_{\msf{Cl}}) = & \bigcup_{j \in \Inter{s}} b \cdot b \cdot a^j \cdot b \cdot b^d \cdot \left( \{ a^{d(\nu(x_{i_j}))}\} \; \cup \; a^{d+1} \cdot U_{k+1}(a) \right)\\
		\cup & \bigcup_{j \in \Inter{s}} b \cdot b \cdot a^j \cdot b \cdot b^d \cdot \left(\{b^{j}\} \cup b^{s + s + T \cdot r} \cdot U_{k+1}(b)\right) \\
		\cup & \bigcup_{j \in \Inter{s}} b \cdot b \cdot a^j \cdot b \cdot b^d \cdot \{b^\sigma \mid \sigma \in \msf{Ind}(x_{i_j}) \setminus \{j\}\}
	\end{align*}
	Furthermore, for all $j \in \Inter{s}$, we have:
	\begin{equation*}
		b \cdot b \cdot a^j \cdot b \cdot b^d \cdot \left(\{b^{j}\} \cup b^{s + s + T \cdot r} \cdot U_{k+1}(b)\right) = \mathcal{P}_{\msf{Cl}}(X,\mathcal{C},k,d,T,j)
	\end{equation*}
	In addition, since $\nu$ satisfies $(X,\mathcal{C})$, if $C_j \in \mathcal{C}_+$, then $\nu(x_{i_j}) = \top$, and if $C_j \in \mathcal{C}_-$, then $\nu(x_{i_j}) = \bot$. Thus, by definition of $d(\nu(x_{i_j}))$ and $\mathcal{P}^{\mathsf{acc}}_{\msf{Cl}}(k,d,j)$, we have $b \cdot b \cdot a^j \cdot b \cdot b^d \cdot a^{d(\nu(x_{i_j}))} \in \mathcal{P}^{\mathsf{acc}}_{\msf{Cl}}(k,d,j)$. Therefore:
	\begin{equation*}
		b \cdot b \cdot a^j \cdot b \cdot b^d \cdot \left( \{ a^{d(\nu(x_{i_j}))}\} \; \cup \; a^{d+1} \cdot U_{k+1}(a) \right) = \mathcal{P}^{\mathsf{acc}}_{\msf{Cl}}(k,d,j)
	\end{equation*}
	Hence, letting $\msf{ErrCl} := \bigcup_{j \in \Inter{s}} b\cdot b\cdot a^j \cdot b \cdot b^d \cdot \{b^\sigma \mid \sigma \in \msf{Ind}(x_{i_j}) \setminus \{j\}\}$, we have:
	\begin{equation}
		\label{eqn:bb}
		\Lan(\mathcal{A}) \cap b \cdot b \cdot \Sigma^* = b \cdot b \cdot \Lan(\mathcal{A},q_{\msf{Cl}}) = \bigcup_{j \in \Inter{s}} \left(\mathcal{P}^{\mathsf{acc}}_{\msf{Cl}}(k,d,j) \cup \mathcal{P}_{\msf{Cl}}(X,\mathcal{C},k,d,T,j)\right) \cup \msf{ErrCl}
	\end{equation}
	Overall, with Equations~(\ref{eqn:aa}),~(\ref{eqn:ab}),~(\ref{eqn:ba}), and~(\ref{eqn:bb}), we obtain:
	\begin{align*}
		\Lan(\mathcal{A}) & = \left(\Lan(\mathcal{A}) \cap a \cdot a \cdot \Sigma^*\right) \cup \left(\Lan(\mathcal{A}) \cap a \cdot b \cdot \Sigma^*\right) \cup \left(\Lan(\mathcal{A}) \cap b \cdot a \cdot \Sigma^*\right) \cup \left(\Lan(\mathcal{A}) \cap b \cdot b \cdot \Sigma^*\right)\\
		& = \mathcal{P}(X,\mathcal{C},k,d,T,M) \cup \msf{ErrVar} \cup \msf{ErrCl}
	\end{align*}
	with 
	\begin{equation*}
		|\msf{ErrVar}| = |\bigcup_{i \in \Inter{r}} b \cdot a \cdot U_M(b) \cdot b^i \cdot a \cdot b^{d} \cdot a^{d(\nu(x_i))}| = M \cdot r
	\end{equation*}
	and, since $\mathsf{App}(x_{i_j}) \cup \mathsf{NotApp}(x_{i_j}) = \Inter{s}$:
	\begin{align*}
		|\msf{ErrCl}| & = |\bigcup_{j \in \Inter{s}} b\cdot b\cdot a^j \cdot b \cdot b^d \cdot \{b^\sigma \mid \sigma \in \msf{Ind}(x_{i_j}) \setminus \{j\}\}| = |\msf{Ind}(x_{i_j})| - 1 \\
		& = \sum_{j \in \Inter{s}} (|\mathsf{App}(x_{i_j}) \cup \{ s + \sigma \mid \sigma \in \mathsf{NotApp}(x_{i_j})\} \cup \{ s + s + i_j + t \cdot r \mid t \in \llbracket 0,T-1 \rrbracket \})\}| -1) \\
		& = \sum_{j \in \Inter{s}} (|\mathsf{App}(x_{i_j})| + |\mathsf{NotApp}(x_{i_j})| + |\llbracket 0,T-1 \rrbracket| -1) \\
		& = s \cdot (s + T -1)
	\end{align*}
	Therefore, we have:
	\begin{itemize}
		\item $|Q| \leq n$;
		\item $\mathcal{P}(X,\mathcal{C},k,d,T,M) \subseteq \Lan(\mathcal{A})$;
		\item $|\Lan(\mathcal{A}) \setminus \mathcal{P}(X,\mathcal{C},k,d,T,M)| = |\msf{ErrCl}| + 	|\msf{ErrVar}| = M \cdot r + s \cdot (s + T -1) \leq k$
	\end{itemize}
	Thus, for all $m \in \N$, there is a  $(\mathcal{P}(X,\mathcal{C},k,d,T,M),n,k,m)$-suitable DFA.
\end{proof}

\subsection{If there is a suitable DFA, then there is a satisfying valuation}
In all of this subsection, we fix a positive instance $(X,\mathcal{C})$ of Problem 2, and $n,k,m,d,M,T \in \N$. We also let $\mathcal{P} := \mathcal{P}(X,\mathcal{C},d,k,T,M)$. The goal of this subsection is to establish that, under some conditions on the integers $n,k,m,d,M,T \in \N$, if there is a $(\mathcal{P},n,k,m)$-suitable DFA, then $(X,\mathcal{C})$ is a positive instance of Problem 2. Let us introduce below the inequalities between the integers that we will consider.
\begin{definition}
	\label{def:inequalities}
	We will consider four different inequalities between the above integers:
	\begin{itemize}
		\item Assumption A: $m \geq 2\max_{u \in \mathcal{P}} |u| + \max(r,d)$;
		\item Assumption B: $k < M \cdot T$;
		\item Assumption C: $
		n < d \cdot (2+r) + d$;
		\item Assumption D: $k \leq s \cdot (T + s-1) + M \cdot r$.
	\end{itemize}
\end{definition}

Let us now state the result that we establish in this subsection.
\begin{lemma}
	\label{lem:assumption_DFA_valuation}
	If Assumptions A, B, C, and D hold, and there is a $(\mathcal{P},n,k,m)$-suitable DFA, then $(X,\mathcal{C})$ is a positive instance of Problem 2.
\end{lemma}

To establish this lemma, we introduce several notations below.
\begin{definition}
	We define two sets of prefixes of words in $\mathcal{P}$. Specifically, we let:
	\begin{align*}
		\msf{Cont}_{k+1}(a) := \{ & \; a \cdot a \cdot u_a, a \cdot b \cdot u_a, b \cdot a \cdot u \cdot b^i \cdot a \cdot b^d \\
		\mid & \; u_a \in U_{k+1}(a), u \in U_M(b), i \in \Inter{r} \}
	\end{align*}
	We consider these prefixes because, by definition, for all $v \in \msf{Cont}_{k+1}(a)$ and $u_a' \in U_{k+1}(a)$, we have: $v \cdot a^{d+1} \cdot u_a' \in \mathcal{P}$. Similarly, we also let:
	\begin{align*}
		\msf{Cont}_{k+1}(b) := \{ & \; b \cdot a \cdot u \cdot b^i \cdot a, b \cdot b \cdot a^j \cdot b \\ \mid & \; u \in U_M(b), i \in \Inter{r}, j \in \Inter{s} \}
	\end{align*}
	We consider these prefixes because, by definition, for all $v \in \msf{Cont}_{k+1}(b)$ and $u_b' \in U_{k+1}(b)$, we have: $v \cdot b^{d+s+s+T \cdot r} \cdot u_b' \in \mathcal{P}$.
\end{definition}

Now, we state a lemma* divided into four consecutive steps (which we prove separately afterwards) from which we will be able to straightforwardly deduce Lemma*~\ref{lem:assumption_DFA_valuation}. 
\begin{lemma}
	\label{lem:central_lemma}
	Assume that there is DFA $\mathcal{A} = (Q,\Sigma,\qinit,\delta,F)$ such that: 1) $|\mathcal{A}| \leq n$; 2) $\mathcal{P} \subseteq \Lan(\mathcal{A})$; and 3) $|\Sigma^{\leq m} \cap (\Lan(\mathcal{A}) \setminus \mathcal{P})| \leq k$. For all words $u,v \in \Sigma^*$, we let:
	\begin{equation*}
		\Delta(u,v) := \{ \delta^*(\qinit,u \cdot w) \mid w \sqsubseteq v, w \neq \epsilon \} \subseteq Q
	\end{equation*}
	
	\textbf{Step 1: Assume that A holds.} For all $l,l' \in \Inter{d}$, we have:
	\begin{itemize}
		\item $\mathsf{Prop}(a,a)$: For all $v_a,v_a' \in \msf{Cont}_{k+1}(a)$, if $l \neq l'$ then: $\delta^*(\qinit,v_a\cdot a^l) \neq \delta^*(\qinit,v_a'\cdot a^{l'})$
		\item $\mathsf{Prop}(b,b)$: For all $v_b,v_b' \in \msf{Cont}_{k+1}(b)$, if $l \neq l'$ then: $\delta^*(\qinit,v_b \cdot b^l) \neq \delta^*(\qinit,v_b' \cdot b^{l'})$
		\item $\mathsf{Prop}(a,b)$: For all $v_a \in \msf{Cont}_{k+1}(a)$, $v_b \in \msf{Cont}_{k+1}(b)$: $\delta^*(\qinit,v_a\cdot a^l) \neq \delta^*(\qinit,v_b \cdot b^{l'})$
	\end{itemize}
	
	Furthermore, there is $u_a^\top,u_a^\bot \in U_{k+1}(a)$ such that $\Delta(a \cdot a \cdot u_a^\top,a^d) \cap \Delta(a \cdot b \cdot u_a^\bot,a^d) = \emptyset$ and
	:
	\begin{align*}
		a \cdot a \cdot u_a^\top \cdot a^d,a \cdot b \cdot u_a^\bot \cdot a^{d+1} \in \Lan(\mathcal{A}) \\
		a \cdot b \cdot u_a^\bot \cdot a^{d},a \cdot a \cdot u_a^\top \cdot a^{d+1} \notin \Lan(\mathcal{A})
	\end{align*}
	
	\textbf{Step 2: Additionally assume that B holds.} For all $i \in \Inter{r}$, there is some $u_i \in U_M(b)$ such that the set $\bigcup_{1 \leq i \leq r} \Delta(b \cdot a \cdot u_i \cdot b^i \cdot a,b^d)$ is of cardinality $r \cdot d$. 
	
	Therefore, the set $V$ defined below is of cardinality $(2 + r) \cdot d$:
	\begin{equation*}
		V := \Delta(a \cdot a \cdot u_a^\top,a^d) \cup \Delta(a \cdot b \cdot u_a^\bot,a^d) \cup \bigcup_{1 \leq i \leq r} \Delta(b \cdot a \cdot u_i \cdot b^i \cdot a,b^d)
	\end{equation*}
	
	\textbf{Step 3: Additionally assume that C holds.} Let $t_\top := \delta^*(\qinit,a \cdot a \cdot u_a^\top \cdot a^d) \in Q$ and $t_\bot := \delta^*(\qinit,a \cdot b \cdot u_a^\bot \cdot a^d)$. 
	\begin{itemize}
		\item Let $i \in \Intr$. For all $u \in U_M(b)$, we have: $\delta^*(\qinit,b \cdot a \cdot u \cdot b^i \cdot a \cdot b^d \cdot a^d) \in \{t_\top,t_\bot \}$. 
		
		Then, letting $t_{x_i} := \delta^*(\qinit,b \cdot a \cdot u_i \cdot b^i \cdot a \cdot b^d)$, we have: $\delta^*(t_{x_i},a^d) \in \{t_\top,t_\bot \}$;
		\item For all $j \in \Ints$, there is some $i_j \in \Intr$ such that: $\delta^*(\qinit,b \cdot b \cdot a^j \cdot b \cdot b^d) = t_{x_{i_j}}$.
	\end{itemize}
	
	\textbf{Step 4: Additionally assume that D holds.} For all $j \in \Ints$, we have $x_{i_j} \in C_j$. Thus, the valuation $\nu: X \to \{\top,\bot\}$ such that, for all $x \in X$, $\delta^*(t_x,a^d) = t_{\nu(x)}$, satisfies $(X,\mathcal{C})$.
\end{lemma}
We prove the four steps of this lemma* in four different proofs. 
Recalling the bird's eye view that we have presented in the paper, the proofs of Steps 1 and 2 rely on TME-arguments, the proof of Step 3 relies on TMS-arguments, and the proof of Step 4 relies on a precise TME-argument.

\paragraph{Proof of Step 1 of Lemma*~\ref{lem:central_lemma}, assuming A}
\begin{proof} Let $l,l' \in \Inter{d}$. The cases $\msf{Prop}(a,a)$ and $\msf{Prop}(b,b)$ are almost identical, and are both very close to the case $\msf{Prop}(a,b)$.
	\begin{itemize}
		\item $\msf{Prop}(a,a)$: Consider some $v_a,v_a' \in \msf{Cont}_{k+1}(a)$. Assume towards a contradiction that $j := l' - l > 0$ and $\delta^*(\qinit,v_a \cdot a^l) = \delta^*(\qinit,v_a' \cdot a^{l'})$. Then, for all $u_a \in U_{k+1}(a)$, we have $\delta^*(\qinit,v_a \cdot a^{d+1-j} \cdot u_a) = \delta^*(\qinit,v_a' \cdot a^{d+1} \cdot u_a) \in F$ since $v_a' \cdot a^{d+1} \cdot u_a \in \mathcal{P} \subseteq \Lan(\mathcal{A})$. In addition, $|v_a \cdot a^{d+1-j} \cdot u_a| \leq |v_a' \cdot a^{d+1} \cdot u_a| \leq m$. Thus: $v_a \cdot a^{d+1-j} \cdot u_a \in \Sigma^{\leq m} \cap (\Lan(\mathcal{A}) \setminus \mathcal{P})$. Hence the contradiction since this holds for all $u_a \in U_{k+1}(a)$ and $|U_{k+1}(a)| = k+1$. 
		\item $\msf{Prop}(b,b)$: Consider some $v_b,v_b' \in \msf{Cont}_{k+1}(b)$. Assume towards a contradiction that $j := l' - l > 0$ and $\delta^*(\qinit,v_b \cdot b^l) = \delta^*(\qinit,v_b' \cdot b^{l'})$. Then, for all $u_b \in U_{k+1}(b)$, we have $\delta^*(\qinit,v_b \cdot b^{d-j+s+s+T\cdot r} \cdot u_b) = \delta^*(\qinit,v_b' \cdot b^{d+s+s+T\cdot r} \cdot u_b) \in F$ since $v_b \cdot b^{d-j+s+s+T\cdot r} \cdot u_b \in \mathcal{P} \subseteq \Lan(\mathcal{A})$. In addition, $|v_b \cdot b^{d-j+s+s+T\cdot r} \cdot u_b| \leq |v_b' \cdot b^{d+s+s+T\cdot r} \cdot u_b| \leq m$. Thus: $v_b \cdot b^{d-j+s+s+T\cdot r} \cdot u_b \in \Sigma^{\leq m} \cap (\Lan(\mathcal{A}) \setminus \mathcal{P})$. Hence the contradiction since this holds for all $u_b \in U_{k+1}(b)$ and $|U_{k+1}(b)| = k+1$. 
		\item $\msf{Prop}(a,b)$: Let $v_a \in \msf{Cont}_{k+1}(a)$ and $v_b \in \msf{Cont}_{k+1}(b)$. Assume towards a contradiction that: $\delta^*(\qinit,v_a \cdot a^l) = \delta^*(\qinit,v_b \cdot b^{l'})$. 	Then, for all $u_b \in U_{k+1}(b)$, we have: $\delta^*(\qinit,v_a \cdot a^l \cdot b^{d-l'+s+s+T \cdot r} \cdot u_b) = \delta^*(\qinit,v_b \cdot b^{d+s+s+T \cdot r} \cdot u_b) \in F$. Furthermore, we have $v_a \cdot a^l \cdot b^{d-l'+s+s+r \cdot T} \cdot u_b \notin \mathcal{P}$, since $l \geq 1$, and $|v_a \cdot a^l \cdot b^{d-l'+s+s+r \cdot T} \cdot u_b| \leq |v_b \cdot b^{d+s+s+T \cdot r} \cdot u_b| + |v_a| - |v_b| + l \leq m$. Thus: $v_a \cdot a^l \cdot b^{d-l'+s+s+r \cdot T} \cdot u_b \in \Sigma^{\leq m} \cap (\Lan(\mathcal{A}) \setminus \mathcal{P})$. Hence, the contradiction since this holds for all $u_b \in U_{k+1}(b)|$ and $|U_{k+1}(b)| = k+1$. 
	\end{itemize}
	
	Furthermore, for all $u_a \in U_{k+1}(a)$, we have $a \cdot a \cdot u_a \cdot a^{d+1},a \cdot b \cdot u_a \cdot a^{d} \in \Sigma^{\leq m} \setminus \mathcal{P}$. Hence, since $|U_{k+1}(a)| = k+1$, it follows that there is some $u_a^\top,u_a^\bot \in U_{k+1}(a)$ such that $a \cdot a \cdot u_a^\top \cdot a^{d+1},a \cdot b \cdot u_a^\bot \cdot a^{d} \notin \Lan(\mathcal{A})$. They are such that $a \cdot a \cdot u_a^\top \cdot a^d,a \cdot b \cdot u_a^\bot \cdot a^{d+1} \in \mathcal{P} \subseteq \Lan(\mathcal{A})$.
\end{proof}

\paragraph{Proof of Step 2 of Lemma*~\ref{lem:central_lemma}, assuming A, B}
\begin{proof}
	Let $i \in \Inter{r}$. Assume towards a contradiction that, for all $u \in U_M(b)$, there is some $v_u \in U_M(b)$ and $i_u \neq i \in \Inter{r}$ such that $\delta^*(\qinit,b \cdot a \cdot u \cdot b^i \cdot a \cdot b^d) = \delta^*(\qinit,b \cdot a \cdot v_u \cdot b^{i_u} \cdot a \cdot b^d)$. Consider some $t \in \llbracket 0,T-1 \rrbracket$ and $u \in U_M(b)$. Since:
	\begin{equation*}
		b \cdot a \cdot u \cdot b^i \cdot a \cdot b^{d+s+s+i+t \cdot r} \in \mathcal{P} \subseteq \Lan(\mathcal{A})
	\end{equation*}
	and $|b \cdot a \cdot v_u \cdot b^{i_u} \cdot a \cdot b^{d+s+s+i+t \cdot r}| \leq |b \cdot a \cdot u \cdot b^i \cdot a \cdot b^{d+s+s+i+t \cdot r}| + |v_u| + i_u \leq m$, we have:
	\begin{equation*}
		b \cdot a \cdot v_u \cdot b^{i_u} \cdot a \cdot b^{d+s+s+i+t \cdot r} \in \Sigma^{\leq m} \cap (\Lan(\mathcal{A}) \setminus \mathcal{P})
	\end{equation*}
	Hence: $\{ b \cdot a \cdot v_u \cdot b^{i_u} \cdot a \cdot b^{d+s+s+i+t \cdot r} \mid u \in U_M(b), t \in \llbracket 0,T-1 \rrbracket \} \subseteq (\Sigma^{\leq m} \cap \Lan(\mathcal{A}) \setminus \mathcal{P})$. Since $|U_M(b)| \cdot |\llbracket 0,T-1 \rrbracket| = M \cdot T > k$, this is a contradiction. Therefore, there is some $u_i \in U_M(b)$, such that, for all $i' \neq i \in \Inter{r}$ and $u \in U_M(b)$, we have $\delta^*(\qinit,b \cdot a \cdot u_i \cdot b^i \cdot a \cdot b^d) \neq \delta^*(\qinit,b \cdot a \cdot u \cdot b^{i'} \cdot a \cdot b^d)$.
	
	Then, by definition, we have:
	\begin{equation*}
		\bigcup_{i \in \Inter{r}} \Delta(b \cdot a \cdot u_i \cdot b^i \cdot a,b^d) = \{ \delta^*(\qinit,b \cdot a \cdot u_i \cdot b^i \cdot a \cdot b^l) \mid i \in \Inter{r},\; l \in \Inter{d} \}
	\end{equation*}
	
	Consider any $i,i' \in \Inter{r}, l,l' \in \Inter{d}$. If $l \neq l'$, then we have $\delta^*(\qinit,b \cdot a \cdot u_i \cdot b^i \cdot a \cdot b^l) \neq \delta^*(\qinit,b \cdot a \cdot u_{i'} \cdot b^{i'} \cdot a \cdot b^{l'})$ by $\msf{Prop}(b,b)$. Furthermore, if $l = l'$ and $\delta^*(\qinit,b \cdot a \cdot u_i \cdot b^i \cdot a \cdot b^l) = \delta^*(\qinit,b \cdot a \cdot u_{i'} \cdot b^{i'} \cdot a \cdot b^{l'})$, then it follows that $\delta^*(\qinit,b \cdot a \cdot u_i \cdot b^i \cdot a \cdot b^d) = \delta^*(\qinit,b \cdot a \cdot u_{i'} \cdot b^{i'} \cdot a \cdot b^{d})$, and thus $i = i'$, by choice of $u_i,u_{i'}$. Overall, we have that if $\delta^*(\qinit,b \cdot a \cdot u_i \cdot b^i \cdot a \cdot b^l) = \delta^*(\qinit,b \cdot a \cdot u_{i'} \cdot b^{i'} \cdot a \cdot b^{l'})$, then $i = i'$ and $l = l'$. Therefore, the set $\bigcup_{i \in \Inter{r}} \Delta(b \cdot a \cdot u_i \cdot b^i \cdot a,b^d)$ is of cardinality $r \cdot d$. 
	
	Then, the set $V$ is by definition equal to:
	\begin{equation*}
		V = \Delta(a \cdot a \cdot u_a^\top,a^d) \cup \Delta(a \cdot b \cdot u_a^\bot,a^d) \cup  \bigcup_{i \in \Inter{r}} \Delta(b \cdot a \cdot u_i \cdot b^i \cdot a,b^d) \subseteq Q
	\end{equation*}
	By $\msf{Prop}(a,a)$, for all $v,v' \in \{a\cdot a \cdot u_a^\top,a\cdot b\cdot u_a^\bot\}$ and $l,l' \in \Inter{d}$, we have $\delta^*(\qinit,v \cdot a^l) = \delta^*(\qinit,v' \cdot a^{l'})$ implies $l = l'$, which implies $\delta^*(\qinit,v \cdot a^d) = \delta^*(\qinit,v' \cdot a^d)$. Since, by choice of $u_a^\top,u_a^\bot$, we have $a\cdot a \cdot u_a^\top \cdot a^d \in \Lan(\mathcal{A})$ and $a\cdot b \cdot u_a^\bot \cdot a^d \notin \Lan(\mathcal{A})$, it follows that $\delta^*(\qinit,v \cdot a^l) = \delta^*(\qinit,v' \cdot a^{l'})$ implies $l = l'$ and $v = v'$. Therefore, the set $\Delta(a \cdot a \cdot u_a^\top,a^d) \cup \Delta(a \cdot b \cdot u_a^\bot,a^d)$ is of cardinality $2 \cdot d$. Furthermore, by $\msf{Prop}(a,b)$, we have:
	\begin{equation*}
		\left(\Delta(a \cdot a \cdot u_a^\top,a^d) \cup \Delta(a \cdot b \cdot u_a^\bot,a^d)\right) \cap \left(\bigcup_{i \in \Inter{r}} \Delta(b \cdot a \cdot u_i \cdot b^i \cdot a,b^d)\right) = \emptyset
	\end{equation*}
	Therefore, the cardinal of the set $V$ is equal to $(2 + r) \cdot d$.
\end{proof}

\paragraph{Proof of Step 3 of Lemma*~\ref{lem:central_lemma}, assuming A, B, C}
\begin{proof} 
	The proofs of both items are very similar, and rather straightforward with the help of $\msf{Prop}(a,a)$, $\msf{Prop}(b,b)$, and $\msf{Prop}(a,b)$.
	
	Consider first some i $\in \Inter{r}$ and $u \in U_M(b)$. Let $W := \Delta(b \cdot a \cdot u \cdot b^i \cdot a \cdot b^d,a^d) \subseteq Q$. By $\msf{Prop}(a,a)$, $W$ is of cardinal $d$. Thus, we have $|Q| \leq n < (2+r) \cdot d + d = |V| + |W|$. Hence, it must be that $V \cap W \neq \emptyset$. That is, there exist $(v,c) \in \{ (a \cdot a \cdot u_a^\top,a),(a \cdot b \cdot u_a^\bot,a),(b \cdot a \cdot u_{i'} \cdot b^{i'} \cdot a,b) \mid i' \in \Inter{r}\}$ and $l,l' \in \Inter{d}$ such that $\delta^*(\qinit,v \cdot c^l) = \delta^*(\qinit,b \cdot a \cdot u \cdot b^i \cdot a \cdot b^d \cdot a^{l'})$. By $\msf{Prop}(a,b)$, it must be that $v \in \{ a \cdot a \cdot u_a^\top,a \cdot b \cdot u_a^\bot \}$. By $\msf{Prop}(a,a)$, it must be that $l = l'$. Hence, we have:
	\begin{equation*}
		\delta^*(\qinit,b \cdot a \cdot u \cdot b^i \cdot a \cdot b^d \cdot a^{l}) \in \{\delta^*(\qinit,a \cdot a \cdot u_a^\top \cdot a^l),\delta^*(\qinit,a \cdot b \cdot u_a^\bot \cdot a^l)\}
	\end{equation*}
	Thus, we have:
	\begin{equation*}
		\delta^*(\qinit,b \cdot a \cdot u \cdot b^i \cdot a \cdot b^d \cdot a^{d}) \in \{\delta^*(\qinit,a \cdot a \cdot u_a^\top \cdot a^d),\delta^*(\qinit,a \cdot b \cdot u_a^\bot \cdot a^d)\} = \{ t_\top,t_\bot\}
	\end{equation*}
	
	Now, consider some $j \in \Ints$ and let $W := \Delta(b \cdot b \cdot a^j \cdot b,b^d) \subseteq Q$. By $\msf{Prop}(b,b)$, $W$ is of cardinal $d$. Thus, we have $|Q| \leq n < (2+r) \cdot d + d = |V| + |W|$. Hence, it must be that $V \cap W \neq \emptyset$. That is, there exist $(v,c) \in \{ (a \cdot a \cdot u_a^\top,a),(a \cdot b \cdot u_a^\bot,a),(b \cdot a \cdot u_{i} \cdot b^{i} \cdot a,b) \mid i \in \Inter{r}\}$ and $l,l' \in \Inter{d}$ such that $\delta^*(\qinit,v \cdot c^l) = \delta^*(\qinit,b \cdot b \cdot a^j \cdot b \cdot b^{l'})$. By $\msf{Prop}(a,b)$, it must be that $v \in \{ b \cdot a \cdot u_{i} \cdot b^{i} \cdot a \mid i \in \Inter{r} \}$. By $\msf{Prop}(b,b)$, it must be that $l = l'$. It follows that there is some $i_j \in \Inter{r}$ such that:
	\begin{equation*}
		\delta^*(\qinit,b \cdot a \cdot u_{i_j} \cdot b^{i_j} \cdot a \cdot b^l) = \delta^*(\qinit,b \cdot b \cdot a^j \cdot b \cdot b^{l})
	\end{equation*}
	Thus:
	\begin{equation*}
		\delta^*(\qinit,b \cdot a \cdot u_{i_j} \cdot b^{i_j} \cdot a \cdot b^d) = \delta^*(\qinit,b \cdot b \cdot a^j \cdot b \cdot b^{d})
	\end{equation*}
	Since $t_{x_{i_j}} = \delta^*(\qinit,b \cdot a \cdot u_{i_j} \cdot b^{i_j} \cdot a \cdot b^d)$, we obtain $	t_{x_{i_j}} = \delta^*(\qinit,b \cdot b \cdot a^j \cdot b \cdot b^{d})$.
\end{proof}

\paragraph{Proof of Step 4 of Lemma*~\ref{lem:central_lemma}, assuming A, B, C, D}
\begin{proof}
	First of all, note that we have:
	\begin{align*}
		t_\top,\delta(t_\bot,a) \in F & \text{ since } a \cdot a \cdot u_a^\top \cdot a^d,a \cdot b \cdot u_a^\bot \cdot a^{d+1} \in \Lan(\mathcal{A}) \\
		t_\bot,\delta(t_\top,a) \notin F & \text{ since } a \cdot b \cdot u_a^\bot \cdot a^d,a \cdot a \cdot u_a^\top \cdot a^{d+1} \notin \Lan(\mathcal{A})
	\end{align*}
	Therefore, we have $t_\top \neq t_\bot$. 
	
	Now, let $i \in \Inter{r}$ and $u \in U_M(b)$. We have shown that:
	\begin{equation*}
		\delta^*(\qinit,b \cdot a \cdot u \cdot b^i \cdot a \cdot b^d \cdot a^d) \in \{t_\top,t_\bot \}
	\end{equation*}
	We let $\msf{U}(i,\top) := \{ u \in U_M(b) \mid \delta^*(\qinit,b \cdot a \cdot u \cdot b^i \cdot a \cdot b^d \cdot a^d) = t_\top\}$ and $\msf{U}(i,\bot) := \{ u \in U_M(b) \mid \delta^*(\qinit,b \cdot a \cdot u \cdot b^i \cdot a \cdot b^d \cdot a^d) = t_\bot\} = U_M(b) \setminus \msf{U}(i,\top)$. Then, 
	we let:
	\begin{align*}
		\msf{Err}^i_{\top,\bot} := \{ & \; b \cdot a \cdot u_\top \cdot b^{i} \cdot a \cdot b^d \cdot a^d,b \cdot a \cdot u_\bot \cdot b^{i} \cdot a \cdot b^d \cdot a^{d+1} \mid \; u_\top \in \msf{U}(i,\top), u_\bot \in \msf{U}(i,\bot) \}
	\end{align*}
	and we have:
	\begin{equation*}
		\msf{Err}^i_{\top,\bot} \subseteq \Sigma^{\leq m} \cap (\Lan(\mathcal{A}) \setminus \mathcal{P}) \text{ with }|\msf{Err}^i_{\top,\bot}| = M
	\end{equation*}
	
	Consider now any $j \in \Inter{s}$. We have shown that:
	\begin{equation*}
		\delta^*(\qinit,b \cdot b \cdot a^j \cdot b \cdot b^{d}) = \delta^*(\qinit,b \cdot a \cdot u_{i_j} \cdot b^{i_j} \cdot a \cdot b^d) = t_{x_{i_j}}
	\end{equation*}
	Furthermore, we have that for all $j' \in \mathsf{App}(x_{i_j})$: 
	\begin{equation*}
		b \cdot a \cdot u_{i_j} \cdot b^{i_j} \cdot a \cdot b^{d+j'} \in \mathcal{P}
	\end{equation*}
	Thus, $b \cdot b \cdot a^j \cdot b \cdot b^{d+j'} \in \Lan(\mathcal{A})$ with $b \cdot b \cdot a^j \cdot b \cdot b^{d+j'} \in \mathcal{P}$ if and only if $j' = j$. Similarly, we have that for all $j' \in \mathsf{NotApp}(x_{i_j})$: 
	\begin{equation*}
		b \cdot a \cdot u_{i_j} \cdot b^{i_j} \cdot a \cdot b^{d+s+j'} \in \mathcal{P} \text{ thus }b \cdot b \cdot a^j \cdot b \cdot b^{d+s+j'} \in \Lan(\mathcal{A}) \setminus \mathcal{P}
	\end{equation*}
	In addition, for all $t \in \llbracket 0,T-1 \rrbracket$, we have:
	\begin{equation*}
		b \cdot a \cdot u_{i_j} \cdot b^{i_j} \cdot a \cdot b^{d+s+s+i_j+t \cdot r} \in \mathcal{P} \text{ thus }b \cdot b \cdot a^j \cdot b \cdot b^{d+s+s+i_j+t \cdot r} \in \Lan(\mathcal{A}) \setminus \mathcal{P}
	\end{equation*}
	Therefore, we let: 
	\begin{align*}
		\msf{ErrApp}^j := b \cdot b \cdot a^j \cdot b \cdot b^d \cdot \{ & \; b^{j_1}, b^{s+j_2}, b^{s+s+i_j+t \cdot r} \\
		\mid & \; j_1 \in \mathsf{App}(x_{i_j}) \setminus \{j\}, j_2 \in \mathsf{NotApp}(x_{i_j}), t \in \llbracket 0,T-1 \rrbracket \}
	\end{align*}
	Since $\mathsf{App}(x_{i_j}) \cup \mathsf{NotApp}(x_{i_j}) = \Inter{s}$, we have:
	\begin{equation*}
		\msf{ErrApp}^j \subseteq \Sigma^{\leq m} \cap (\Lan(\mathcal{A}) \setminus \mathcal{P}) \text{ with }|\msf{ErrApp}^j| \geq T + s-1
	\end{equation*}
	and the above inequality is an equality if and only if $j \in \msf{App}(x_{i_j})$. Overall, we have:
	\begin{equation*}
		\left(\bigcup_{i \in \Inter{r}} \msf{Err}^i_{\top,\bot} \cup \bigcup_{j \in \Inter{s}} \msf{ErrApp}^j\right) \subseteq \Sigma^{\leq m} \cap (\Lan(\mathcal{A}) \setminus \mathcal{P})
	\end{equation*}
	Hence, we have:
	\begin{align*}
		M \cdot r + s \cdot (T + s-1) & \geq  |\Sigma^{\leq m} \cap (\Lan(\mathcal{A}) \setminus \mathcal{P})| \\
		& \geq |\bigcup_{i \in \Inter{r}} \msf{Err}^i_{\top,\bot} \cup \bigcup_{j \in \Inter{s}} \msf{ErrApp}^j | = \sum_{i \in \Inter{r}} |\msf{Err}^i_{\top,\bot}| + \sum_{j \in \Inter{s}} |\msf{ErrApp}^j| \\
		& \geq \sum_{i \in \Inter{r}} M + \sum_{j \in \Inter{s}} (T + s -1) = M \cdot r + s \cdot (T + s -1)
	\end{align*}
	Therefore, all of the above inequalities are in fact equalities. Hence, for all $j \in \Inter{s}$, we have $|\msf{ErrApp}^j| = T + s -1$, that is $j \in \msf{App}(x_{i_j})$. 
	
	We can now conclude. Consider the valuation $\nu: X \to \{\top,\bot\}$ such that, for all $x \in X$, $\delta^*(t_x,a^d) = t_{\nu(x)}$. Consider any $j \in \Inter{s}$. 
	\begin{itemize}
		\item Assume that $C_j \in \mathcal{C}_+$. We have $b \cdot b \cdot a^j \cdot b \cdot b^d \cdot a^{d} \in \mathcal{P}$. Therefore, $\delta^*(t_{x_{i_j}},a^d) = \delta^*(\qinit,b \cdot b \cdot a^j \cdot b \cdot b^d \cdot a^{d}) \in F$. That is $\delta^*(t_{x_{i_j}},a^d) = t_\top$. Hence, $v(x_{i_j}) = \top$ and $x_{i_j} \in C_j$ since $j \in \msf{App}(x_{i_j})$. 
		\item Assume that $C_j \in \mathcal{C}_-$. We have $b \cdot b \cdot a^j \cdot b \cdot b^d \cdot a^{d+1} \in \mathcal{P}$. Therefore, $\delta^*(t_{x_{i_j}},a^{d+1}) = \delta^*(\qinit,b \cdot b \cdot a^j \cdot b \cdot b^d \cdot a^{d+1}) \in F$. That is $\delta^*(t_{x_{i_j}},a^d) = t_\bot$. Hence, $v(x_{i_j}) = \bot$ and $x_{i_j} \in C_j$ since $j \in \msf{App}(x_{i_j})$. 
	\end{itemize}
	In fact, the valuation $\nu$ does satisfy $(X,\mathcal{C})$.
\end{proof}

The proof of Lemma*~\ref{lem:assumption_DFA_valuation} is now direct.
\begin{proof}
	Assuming that A, B, C, and D hold, by Lemma*~\ref{lem:central_lemma}, we can exhibit a valuation $\nu: X \to \{\top,\bot\}$ satisfying $(X,\mathcal{C})$. Thus, $(X,\mathcal{C})$ is a positive instance of Problem 2.
\end{proof}

We can now proceed to the proof of Lemma*~\ref{lem:NP-hard}.
\begin{proof}
	Consider an instance $(X,\mathcal{C})$ of Problem 2. Let $n,k,m,d,M,T \in \N$ and $\mathcal{P} := \mathcal{P}(X,\mathcal{C},d,k,T,M)$. Assume that we have the following inequalities:
	\begin{enumerate}
		\item $n \geq \omega_1(X,d) + \omega_2(X,\mathcal{C},k,T,M)
		$;
		\item $k \geq M \cdot r + s \cdot (s + T -1)$;
		\item $m \geq 2\max_{u \in \mathcal{P}} |u| + \max(r,d)$;
		\item $k < M \cdot T$;
		\item $
		n < d \cdot (2+r) + d = \omega_1(X,d) + d$;
		\item $k \leq s \cdot (T + s-1) + M \cdot r$.
	\end{enumerate}
	
	Then, by Lemmas*~\ref{lem:valuation_implies_automaton} and~\ref{lem:assumption_DFA_valuation}, $(X,\mathcal{C})$ is a positive instance of Problem 2 if and only if there is a $(\mathcal{P},n,k,m)$-suitable DFA. Let us choose these integers $n,k,m,d,M,T \in \N$ such that all of the above inequalities are met. First, we choose:
	\begin{equation*}
		M := 3 (s + r) \text{ and } T = 2 s + 3 r
	\end{equation*}
	and (due to inequalities 2. and 6.):
	\begin{equation*}
		k = s \cdot (T + s-1) + M \cdot r
	\end{equation*}
	That way, we have $k = s \cdot (3 (s + r) - 1) + 3 (s+r) \cdot r \leq 3(s+r)^2 < M \cdot T$. Thus, inequality 4. is also satisfied. Furthermore, we let:
	\begin{equation*}
		d := \omega_2(X,\mathcal{C},k,T,M) + 1
	\end{equation*}
	and 
	\begin{equation*}
		n := \omega_1(X,d) + \omega_2(X,\mathcal{C},k,T,M)
	\end{equation*}
	With these choices, inequalities 1. and 5. are satisfied. 
	
	Let us consider the quantity $\max_{u \in \mathcal{P}} |u|$. Assuming that $r \geq 2$---which is harmless, up to adding a dummy variable which does not appear in any clause in the input $(X,\mathcal{C})$---one can see that, for any $u \in \mathcal{P}$, the set $\Delta(q_0,\epsilon,u)$ constitutes less than half of all the states in the DFA. (We use the assumption $r \geq 2$ for words starting by $b$.) 
	Therefore, we have $\max_{u \in \Pos} |u| \leq |Q|/2 \leq n/2$. Hence, any $m \geq n + \max(r,d) = n + d$ satisfies inequalities 3. This is in particular the case for $m = 2n-2$, since $n \geq 2d$. 
	
	Therefore, considering $k' := k + |\Sigma^{\leq 2n-2} \cap \mathcal{P}(X,\mathcal{C},d,k,T,M)|$ = $k + |\mathcal{P}(X,\mathcal{C},d,k,T,M)|$, we have that $(X,\mathcal{C})$ is a positive instance of Problem 2 if and only if there is a $(\mathcal{P},n,k,2n-2)$-suitable DFA if and only if $(\Pos,n,k')$ is a positive instance of Problem 2 Therefore, Problem 1 is $\msf{NP}$-hard, even with $k$ written in unary, and $\Sigma$ of cardinality 2.
\end{proof}

\section{Our algorithms}
\label{appen:sec3}
\subsection{Our ILP encoding}
Let us first recall the formal setting of Integer Linear Programming problems.
\begin{mpb}
	\label{ILP}
	Consider as input a set $I$ and
	\begin{itemize}
		\item for all $i \in I$, a lower bound $m_i \in \Z$ and an upper bound $M_i \in \Z$ on the value of the integer variable $x_i$; and
		\item an $I$-indexed integer vector $(c_i)_{i \in I} \in \Z^I$; and
		\item a set $\Omega$ indexing $I$-indexed integer vectors $(p_{\omega,i})_{\omega \in \Omega,i \in I}$ and, for each $\omega \in \Omega$, an integer $p_\omega \in \Z$. 
	\end{itemize}
	The goal is to find, over all $I$-indexed integer vectors $(x_i)_{i \in I} \in \Pi_{i \in I} \TInter{m_i,M_i}$, the minimal value of $\sum_{i \in I} c_i \cdot x_i \in \Z$, subject to the inequalities: for all $\omega \in \Omega$, $\sum_{i \in I} p_{\omega,i} \cdot x_i \geq p_{\omega}$.
	
	The size of the input of this problem is equal to $|I| \cdot |\Omega| \cdot \log N$, where $N$ is the maximum (absolute) value of all the input integers.
\end{mpb}

Our goal is, given an alphabet $\Sigma$, a set of word $\Pos \subseteq \Sigma^*$, and an integer $n$, to define a set $I$, an $I$-indexed integer vector $(c_i)_{i \in I} \in \Z^I$, a set $\Omega$ and, for all $\omega \in \Omega$, an $I$-indexed integer vector $(p_{\omega,i})_{i \in I} \in \Z^I$, and an integer $p_\omega \in \Z$ such that:
\begin{itemize}
	\item[$\A \rightarrow x$:] From all DFAs $\A \in \msf{Rec}(\Pos,n)$, we can derive an integer vector $x \in \Pi_{i \in I} \TInter{m_i,M_i}$ satisfying $\sum_{i \in I} p_{\omega,i} \cdot x_i = p_{\omega}$, for all $\omega \in \Omega$, and such that $\sum_{i \in I} c_i \cdot x_i = |\Lan(\A) \cap \Sigma^{\leq 2n-2}|$.
	\item[$x \rightarrow \A$:] Reciprocally, from all integer vectors $x \in \Pi_{i \in I} \TInter{m_i,M_i}$ satisfying $\sum_{i \in I} p_{\omega,i} \cdot x_i = p_{\omega}$, for all $\omega \in \Omega$, we can (easily) derive a DFA $\A \in \msf{Rec}(\Pos,n)$ such that $\sum_{i \in I} c_i \cdot x_i = |\Lan(\A) \cap \Sigma^{\leq 2n-2}|$.
\end{itemize}

To do so, we encode the prospective DFA with integer variables $(x_i)_{i \in I}$. To this end, we assume that the set of states of the DFA is $Q = \{ q_0, \ldots, q_{n-1} \}$. Then, we use the following binary variables (i.e. integer variables for which the lower bound is 0 and the upper bound is 1) and integer variables. Let us first introduce the binary variables, their informal meaning, and the inequalities that we define to ensure their meaning.
\begin{description}[font=\normalfont\itshape]
	\item[Transition variables.]
	For each potential transition $\delta(p, \alpha) = q$ (for $p,q \in Q$, $\alpha \in \Sigma$), we introduce a binary variable $x_{p, \alpha, q}$. The meaning of $x_{p, \alpha, q} = 1$ is $\delta(p,\alpha) = q$ in the prospective DFA. We let $I_{\msf{Trans}} := Q \times \Sigma \times Q$.
	\item[State variables.]
	For each state $q \in Q$, we introduce a binary variable $x_q$. The meaning of $x_q = 1$ is that $q \in F$. We let $I_{\msf{State}} := Q$.
	\item[Word variables.]
	For each state $q \in Q$, and prefixes $u \in \Sigma^*$ of words in $\Pos$, we introduce a binary variable $x_{u,q}$. The meaning of $x_{u,q} = 1$ is that $\delta^*(q_0,u) = q$. We let $I_{\msf{Word}} := \msf{Pref}(\Pos) \times Q$.
\end{description}

Now, we impose constraints on these variables $(x_i)_{i \in I}$ to enforce their meaning by defining the various (in)equalities, indexed by the set $\Omega$, that must be satisfied. Each one of these (in)equalities, which corresponds to a different element of $\omega \in \Omega$, implicitly defines the integer constant $p_\omega \in \Z$ and integer vector $(p_{\omega,i})_{i \in I} \in \Z^I$.

We first impose that these variables encode a DFA, that is: for all $p \in Q$ and $\alpha \in \Sigma$, there is unique $q \in Q$ such that $\delta(p,\alpha) = q$. This is ensured by the following equalities:
\begin{equation}
	\forall p \in Q,\; \forall \alpha \in \Sigma: \; \sum_{q \in Q} x_{p, \alpha, q} = 1
	\tag{$\omega_{\msf{DFA}}^{p,\alpha}$}
\end{equation}
We let $\Omega_{\msf{DFA}} := \{\omega_{\msf{DFA}}^{p,\alpha} \mid p \in Q,\; \alpha \in \Sigma\}$.

Next, we need to enforce that the prospective DFA does accept all the words in $\Pos$. To this end, we first enforce the meaning of the state variables by defining (in)equalities corresponding to the following facts: $\delta^*(q_0,\epsilon) = q_0$; for all $u \in \msf{Pref}(\Pos)$, there is a unique $q \in Q$ such that $\delta^*(q_0,u) = q$; and for all $u \cdot \alpha \in \msf{Pref}(\Pos)$, $\delta^*(q,u \cdot \alpha) = \delta(\delta^*(q,u),\alpha)$:
\begin{align}
	& \; x_{\epsilon, q_0} = 1 \tag{$\omega_{\msf{run}}^{\msf{init}}$}\\ 
	\forall u \in \msf{Pref}(\Pos): & \; \sum_{q \in Q} x_{u, q} = 1 \tag{$\omega_{\msf{run}}^u$}\\
	\forall u \cdot \alpha \in \msf{Pref}(\Pos),\; \forall p,q \in Q: & \; x_{u, p} + x_{p, \alpha, q} \leq 1 + x_{u \cdot \alpha, q} \tag{$\omega_{\msf{run}}^{u \cdot \alpha,p,q}$}
\end{align}
The last inequality amounts to the logical implication: ($x_{u, p} = 1$ and $x_{p,\alpha,q} = 1$) implies $x_{u \cdot \alpha, q} = 1$. We let $\Omega_{\msf{run}} := \{ \omega_{\msf{run}}^{\msf{init}},\omega_{\msf{run}}^{u},\omega_{\msf{run}}^{v \cdot \alpha,p,q} \mid u,v \cdot \alpha \in \msf{Pref}(\Pos), p,q \in Q \}$.

We can now enforce that the prospective DFA is consistent with $\Pos$, i.e. for all $u \in \Pos$, we have $\delta^*(q,u) \in F$. This is ensured with the following inequalities:
\begin{equation}
	\forall u \in \Pos,\; \forall q \in Q: \; x_{u, q} \leq x_q \tag{$\omega_{\msf{consistent}}^{u,q}$}
\end{equation}
We let $\Omega_{\msf{consistent}} := \{ \omega_{\msf{consistent}}^{u,q} \mid u \in \Pos, q \in Q \}$.

Let us consider counting variables.
\begin{description}[font=\normalfont\itshape]
	\item[Counting variables.]
	For each state $q \in Q$ and $k \in \TInter{1,2n-2}$, we introduce an integer variable $x_{q,k}$ which counts the number of paths from $q_0$ to $q$ of size exactly $k$ (i.e. $\msf{N}(q,k)$). Thus, we set $m_{q,k} := 0$ and $M_{q,k} := |\Sigma|^k$. 
	
	Since we consider only linear (in)equalities, and thus we cannot multiply two variables, we also consider, for each state $p \in Q$, letter $\alpha \in \Sigma$, state $q \in Q$ and $k \in \TInter{0,2n-3}$, the integer variable $x_{p,\alpha,q,k}$, that stands for the product $x_{p,\alpha,q} \cdot x_{p,k}$. We set $m_{p,\alpha,q,k} := 0$ and $M_{p,\alpha,q,k} := |\Sigma|^k$. We let $I_{\msf{Count}} := Q \times \TInter{0,2n-2} \cup Q \times \Sigma \times Q \times \TInter{0,2n-3}$.
	
	Similarly, since we are interested in the accepted words, i.e., those that lead from the initial state to an accepting state, we also introduce, for each state $q \in Q$ and $k \in \TInter{0,2n-2}$, the variable $x_{q,k,F}$ which stands for the product $x_{q,k} \cdot x_q$ (recall that $x_q = 1$ means that $q \in F$). We set $m_{q,k,F} = 0$ and $M_{q,k,F} = |\Sigma|^k$.
	
	Finally, we introduce an integer variable $x_{F}$ counting the number of accepted words of size at most $2n-2$ (i.e. $|\Lan(\A) \cap \Sigma^{\leq 2n-2}|$). Thus, we set $m_{F} = 0$ and $M_{F} = \sum_{k = 0}^{2n-2} |\Sigma|^k$.  We let $I_{\msf{FinalCount}} := \{F\} \cup Q \times \TInter{0,2n-2} \times \{F\}$.
\end{description}
Overall, the set of variable indices $I$ that we consider is equal to $I := I_{\msf{Trans}} \cup I_{\msf{State}} \cup I_{\msf{Word}} \cup I_{\msf{Count}} \cup I_{\msf{FinalCount}}$. 

The quantity to minimize is the number of accepted words of size at most $2n-2$, i.e., given the meaning of the above variables: $x_F$. Thus, we define $c_F := 1$ and for all $i \in I \setminus \{F\}$, $c_i := 0$.

Let us now enforce the meaning of the counting variables:
\begin{align}
	& \; x_{q_0,0} = 1 \tag{$\omega_{\msf{count}}^{\msf{init},q_0}$}\\ 
	\forall q \in Q \setminus \{q_0\},\; : & \; x_{q,0} = 0 \tag{$\omega_{\msf{count}}^{\msf{init},q}$}\\ 
	\forall (p,\alpha,q) \in Q \times \Sigma \times Q,\; \forall k \in \TInter{0,2n-3},\; : & \;  x_{p,k} \leq x_{p,\alpha,q,k} + (1- x_{p,\alpha,q}) \cdot M \tag{$\omega_{\msf{count},1}^{\msf{step},p,\alpha,q,k}$}\\
	\forall (p,\alpha,q) \in Q \times \Sigma \times Q,\; \forall k \in \TInter{0,2n-3},\; : & \;  x_{p,\alpha,q,k} \leq x_{p,\alpha,q} \cdot M \tag{$\omega_{\msf{count},2}^{\msf{step},p,\alpha,q,k}$}\\
	\forall q \in Q \times \Sigma \times Q,\; \forall k \in \TInter{1,2n-2},\; : & \;  x_{q,k} = \sum_{p \in Q} \sum_{\alpha \in \Sigma} x_{p,\alpha,q,k-1} \tag{$\omega_{\msf{count}}^{\msf{step},q,k}$} \\
	\forall q \in Q,\; \forall k \in \TInter{0,2n-2},\; : & \; x_{q,k} \leq x_{q,k,F} + (1- x_{q}) \cdot M \tag{$\omega_{\msf{count},1}^{\msf{final},q,k}$} \\
	\forall q \in Q,\; \forall k \in \TInter{0,2n-2},\; : & \; x_{q,k,F} \leq  x_{q} \cdot M \tag{$\omega_{\msf{count},2}^{\msf{final},q,k}$}\\ 
	& \; x_{F} = \sum_{q \in Q} \sum_{k \in \TInter{0,2n-2}} x_{q,k,F} \tag{$\omega_{\msf{count}}^{\msf{final}}$}
\end{align}
The constraints $\omega_{\msf{count},1}^{\msf{step},p,\alpha,q,k},\omega_{\msf{count},2}^{\msf{step},p,\alpha,q,k},\omega_{\msf{count},1}^{\msf{final},q,k},\omega_{\msf{count},2}^{\msf{final},q,k}$ use the classical big-$M$ method to encode a product (between a binary and an integer variable) in ILP, which works as soon as $M$ is bigger than the value of any of the variables, e.g. $M = M_F+1$ works. Furthermore, the constraints $\omega_{\msf{count}}^{\msf{init},q},\omega_{\msf{count}}^{\msf{step},q,k}$ exactly follow the iterative computation described in Lemma 1, ensuring that we do have $x_{q,k} = \msf{N}(q,k)$.

We let $\Omega_{\msf{count}} := \{ \omega_{\msf{count}}^{\msf{init},q},\omega_{\msf{count},1}^{\msf{step},p,\alpha,q,k_1},\omega_{\msf{count},2}^{\msf{step},p,\alpha,q,k_1},\omega_{\msf{count}}^{\msf{step},q,k_2},\omega_{\msf{count},1}^{\msf{final},q,k_3},\omega_{\msf{count},2}^{\msf{final},q,k_3},\omega_{\msf{count}}^{\msf{final}} \mid p,q \in Q,\; \alpha \in \Sigma,\; k_1 \in \TInter{0,2n-3},\; k_2 \in \TInter{1,2n-2},\; k_3 \in \TInter{0,2n-2} \}$.

The set $\Omega$ that we consider is then constituted of all the (in)equality indexes defined above:
\begin{equation*}
	\Omega := \Omega_{\msf{DFA}} \cup \Omega_{\msf{run}} \cup \Omega_{\msf{consistent}} \cup \Omega_{\msf{count}}
\end{equation*}

With these constraints, the informal meanings of all these integers variables are ensured, and thus conditions $\A \rightarrow x$ and $x \rightarrow \A$ are met. 

Let us count the number of variables and (in)equalities that we have defined, letting $p := |\msf{Pref}(\Pos)|$ and $s := |\Sigma|$:
\begin{align*}
	|I| & \leq \underbrace{n^2 \cdot s}_{=|I_{\msf{Trans}}|} + \underbrace{n}_{=|I_{\msf{State}}|} +\underbrace{p \cdot n}_{=|I_{\msf{Word}}|} + \underbrace{n \cdot 2n + n^2 \cdot s \cdot 2n}_{\leq|I_{\msf{Count}}|} + \underbrace{1 + n \cdot 2n}_{=|I_{\msf{FinalCount}}|}\\
	& = O(n^3 \cdot s + p \cdot n)
\end{align*}
and 
\begin{align*}
	|\Omega| & \leq \underbrace{n \cdot s}_{=|\Omega_{\msf{DFA}}|} + \underbrace{1 + p + p \cdot n^2}_{\leq|\Omega_{\msf{run}}|} +\underbrace{p \cdot n}_{=|\Omega_{\msf{consistent}}|} + \underbrace{n + 2 \cdot n^2 \cdot s \cdot 2n + n \cdot 2n + 2 \cdot n \cdot 2n + 1}_{\leq|\Omega_{\msf{count}}|} \\
	& = O(n^3 \cdot s + n^2 \cdot p)
\end{align*}

Furthermore, the larger integer defined is $M_F+1$, with $\log(M_F) \leq \log(2n) + k \cdot \log(|\Sigma|)$. Therefore, the size of the input $((c_i)_{i \in I},(p_{\omega,i})_{\omega \in \Omega, i \in I})$ is polynomial in the size of $(\Pos,n)$.  

\subsection{Our heuristic algorithm}
In this subsection, we formally describe our heuristic algorithm. Let us first formally introduce the notion of transition system in the definition* below.
\begin{definition}
	\label{def:TS}
	A \emph{transition system} $T$ is a tuple $T = (Q,\Sigma,\delta)$ where $Q$ is a non-empty set of states, $\Sigma$ is an alphabet, and $\delta\colon Q \times \Sigma \to Q$ is the \emph{transition function}.
	
	Consider any transition system $T = (Q,\Sigma,\delta)$. For all $\mathcal{P} \subseteq \Sigma^*$ and $q \in Q$, we consider the DFA $\mathcal{A}_{T,q,\mathcal{P}} := (Q,\Sigma,\delta,q,\{ \delta^*(q,u) \mid u \in \mathcal{P} \})$ and, for all $n \geq 1$, we let:
	\begin{equation*}
		\msf{score}(T,\mathcal{P},n) := \max_{q \in Q} |\Lan(\mathcal{A}_{T,q,\mathcal{P}}) \cap \Sigma^{\leq 2n-2}|
	\end{equation*} 
	
	Furthermore, for all $(q,\alpha,q') \in Q \times \Sigma \times Q$, we let $T[q,\alpha \to q']$ denote the transition system $T[q,\alpha \to q'] = (Q,\Sigma,\delta')$, with $\delta'\colon Q \times \Sigma \to Q$ is such that, for all $(\tilde{q},\beta) \in Q \times \Sigma$, we have:
	\begin{align*}
		\delta'(\tilde{q},\beta) := \begin{cases}
			q' & \text{ if }(\tilde{q},\beta) = (q,\alpha) \\
			\delta(q,\alpha) & \text{ otherwise }
		\end{cases}
	\end{align*}		
\end{definition}

Our heuristic algorithm is formally described as Algorithm~\ref{algo:heuristic}, where $Q_n$ refers to any set of cardinal $n$. Let us explain step by step this algorithm:
\begin{itemize}
	\item First, in Line 1, we initialize the best total score. We will attempt to improve it with $\msf{NbRun}$ runs, as can be witnessed by the \textquotedblleft{}For\textquotedblright{} loop from Line 2 to Line 24.
	\item Then, in Lines 3-12, we randomly pick a transition system with the best score (as given in Definition*~\ref{def:TS}) over $\msf{InitRand}$ attempts. The $\msf{Rand}$ function used in Line 6 picks uniformly at random in the set $Q_n$ an $\alpha$-successor to the state $q$.
	\item Following, in Lines 13-22, we iteratively improve the score of the current transition system by changing a single transition at a time. With the \textquotedblleft{}For\textquotedblright{} loop of Line 17, we range over all possible transition changes (one at a time), and we pick the transition change that most improves the score. If no change improves the score, the Boolean variable $\msf{HasImproved}$ stays False, and we exit the \textquotedblleft{}While\textquotedblright{} loop.
	\item In Lines 23-24, we update, if relevant, the best score over all the transition systems computed so far on all runs.
	\item Finally, in Lines 25-27, once we have chosen the transition system with the best score over $\msf{NbRun}$ runs, we extract a DFA from that transition system whose number of accepted words (of length at most $2n-2$) is equal to the score of the transition system.
\end{itemize}
Note that the scores in Lines 8 and 18 are computed via matrix multiplication (with a fast exponentiation method).
\begin{algorithm}[t]
	\caption{$\msf{MinScore}(\mathcal{P},\Sigma,n)$: Computes a DFA in $\msf{Rec}(\mathcal{P},n)$ with a score a small as possible.}
	\label{algo:heuristic}
	\textbf{Input}: An alphabet $\Sigma$, a positive set of words $\mathcal{P} \subseteq \Sigma^*$, a bound $n \in \N$ on the number of states
	\begin{algorithmic}[1]
		\State $\msf{TotalBestScore} = \infty$
		\For{$i \in \llbracket 1,\msf{NbRun} \rrbracket$}
		\State $\msf{BestScore} = \infty$
		\For{$j \in \llbracket 1,\msf{InitRand} \rrbracket$}
		\For{$q \in Q_n$, $\alpha \in \Sigma$}
		\State $\delta(q,\alpha) \gets \msf{Rand}(Q_n)$
		\EndFor
		\State $T \gets (Q_n,\Sigma,\delta)$
		\State $s \gets \msf{Score}(T,\mathcal{P},n)$
		\If{$s < \msf{BestScore}$}
		\State $\msf{BestScore} \gets s$, $\msf{BestTS} \gets T$
		\EndIf
		\EndFor
		\State $\msf{CurrentScore} \gets \msf{BestScore}$
		\State $\msf{CurrentTS} \gets \msf{BestTS}$
		\State $\msf{HasImproved} \gets \msf{True}$
		\While{$\msf{HasImproved}$}
		\State $\msf{HasImproved} \gets \msf{False}$
		\For{$(q,\alpha,q') \in Q_n \times \Sigma \times Q_n$}
		\State $T \gets \msf{CurrentTS}[q,\alpha \to q']$
		\State $s \gets \msf{Score}(T,\mathcal{P},n)$
		\If{$s < \msf{CurrentScore}$}
		\State $\msf{CurrentScore} \gets s$, $\msf{CandidateTS} \gets T$, $\msf{HasImproved} \gets \msf{True}$
		\EndIf
		\EndFor
		\If{$\msf{HasImproved}$}
		\State $\msf{CurrentTS} \gets \msf{CandidateTS}$ 
		\EndIf
		\EndWhile
		\If{$\msf{CurrentScore} < \msf{TotalBestScore}$}
		\State $\msf{TotalBestScore} \gets \msf{CurrentScore}$, $\msf{TotalBestTS} \gets \msf{CurrentTS}$
		\EndIf
		\EndFor
		\For{$q \in Q_n$}
		\State $\mathcal{A} \gets \mathcal{A}_{\msf{TotalBestTS},q,\mathcal{P}}$
		\If{$|\Lan(\mathcal{A}) \cap \Sigma^{\leq 2n-2}| = \msf{TotalBestScore}$}
		\Return $\mathcal{A}$
		\EndIf
		\EndFor
	\end{algorithmic}
\end{algorithm}

\end{document}